\documentclass[a4paper, 12pt]{article}

\usepackage[sort&compress]{natbib}
\bibpunct{(}{)}{;}{a}{}{,} 

\usepackage{float}
\usepackage{amsthm, amsmath, amssymb, mathrsfs, multirow, url}
\usepackage{graphicx} 
\usepackage{subcaption}
\usepackage{ifthen} 
\usepackage{amsfonts}
\usepackage[usenames]{color}
\usepackage{fullpage}
\usepackage{textcomp}
\usepackage{multirow}
\usepackage{dsfont}
\usepackage{tabularx}
\usepackage{bigints}
\usepackage[ruled,vlined]{algorithm2e}

\newbox{\bigpicturebox}

\theoremstyle{plain} 
\newtheorem{thm}{Theorem}

\theoremstyle{definition}

\theoremstyle{remark} 

\newtheorem{ex}{Example}

\newcommand{\unif}{{\sf Unif}}
\newcommand{\nm}{{\sf N}}

\newcommand{\gam}{{\sf Gamma}}

\newcommand{\bet}{{\sf Beta}}

\newcommand{\RR}{\mathbb{R}}

\newcommand{\XX}{\mathbb{X}}

\newcommand{\UU}{\mathbb{U}}

\renewcommand{\SS}{\mathbb{S}}

\newcommand{\iid}{\overset{\text{\tiny iid}}{\,\sim\,}}
\newcommand{\ind}{\overset{\text{\tiny ind}}{\,\sim\,}}

\title{A PRticle filter algorithm for nonparametric estimation of multivariate mixing distributions}

\author{Vaidehi Dixit\footnote{Department of Statistics, North Carolina State University; {\tt vdixit@ncsu.edu}, {\tt rgmarti3@ncsu.edu}} \quad and \quad Ryan Martin$^*$}
\date{\today}

\begin{document}

\maketitle 

\begin{abstract}
Predictive recursion (PR) is a fast, recursive algorithm that gives a smooth estimate of the mixing distribution under the general mixture model.  However, the PR algorithm requires evaluation of a normalizing constant at each iteration.  When the support of the mixing distribution is of relatively low dimension, this is not a problem since quadrature methods can be used and are very efficient.  But when the support is of higher dimension, quadrature methods are inefficient and there is no obvious Monte Carlo-based alternative.  In this paper, we propose a new strategy, which we refer to as {\em PRticle filter}, wherein we augment the basic PR algorithm with a filtering mechanism that adaptively reweights an initial set of particles along the updating sequence which are used to obtain Monte Carlo approximations of the normalizing constants. Convergence properties of the PRticle filter approximation are established and its empirical accuracy is demonstrated with simulation studies and a marked spatial point process data analysis.

%can ease calculations in the PR step such that it be feasibly used for estimation of mixing distributions involving more than two variables. We show that the method performs well as compared to PR for univariate and bivariate situations and also gives good results for multivariate mixture modeling.

\smallskip
\emph{Keywords and phrases:} importance sampling; marked point process; mixture model; Monte Carlo; predictive recursion.
\end{abstract}

\section{Introduction}\label{intro}

Suppose we have independent and identically distributed (iid) data $X_1,\ldots,X_n$ having common density $m$ supported on $\XX$.  Furthermore, suppose that we believe this density has the mixture form $m=m_P$, where
\begin{equation}
\label{eq:mix}
    m_{P}(x) = \int_\UU k(x \mid u) \, P(du), \quad x \in \XX,
\end{equation}
with $k(x \mid u)$ a known kernel density and $P$ an unknown mixing distribution supported on $\UU$.  The family in \eqref{eq:mix} indexed by $P$ is commonly referred to as a {\em mixture model}.  One interpretation of the mixture model is that there is a set of underlying latent variables driving the data-generating process.  
%major reason why a mixture model might be adopted is that it is believed there are latent variables underlying or driving the observed process.  
That is, suppose the $X_i$'s are obtained through the two-step process: 
\begin{align*}
U_1,\ldots,U_n & \iid P \\
(X_i \mid U_i) & \ind k(x \mid U_i), \quad i=1,\ldots,n.
\end{align*}
It is easy to check that $X_1,\ldots,X_n$ from this hierarchical model formulation are iid with density $m_P$. This sort of hierarchical, latent variable modeling is common when heterogeneity is present in the observed data.  This also covers the class of problems where $U$ represents an unobservable ``signal'' of interest and $X$ the corresponding noise-corrupted signal, i.e., the ``signal plus noise.'' One also might adopt \eqref{eq:mix} simply for the flexibility the mixture model affords \citep[e.g.,][Chapter 33]{dasgupta2008}. In any case, the distribution of the latent variables, or signals, may be of some practical interest, in which case the goal becomes estimation of the unknown mixing distribution $P$ based on iid data $X_1,\ldots,X_n$ from the mixture $m_P$ in \eqref{eq:mix}.  This is our focus in the present paper.  

%Mixture distributions are powerful and flexible modeling tools for nonparametric density estimation.  Even the very basic class of location-scale can closely approximate {\em any} probability density in the total variation sense \citep[e.g.,][]{dasgupta2008}.  More specifically, suppose we have independent and identically distributed data $X_1,\ldots,X_n$ having common density $m$.  

%Mixing a known kernel density over an unknown mixing distribution creates a robust family of distributions, and this mixture model is fit by estimating the unknown mixing distribution. Suppose we have data $X_1, \ldots, X_n$ from a density $m$, where the goal is estimation of $m$. Then a mixture density of the form $m_{P}$ in \eqref{eq:mix} provides a robust approach to estimation of $m$ with minimal distributional assumptions.

Estimation of the mixing distribution $P$ is a notoriously difficult problem.  Aside from methods tailored to specific mixture model forms, e.g., deconvolution \citep{fan1991, stefanski1990}, there are a few general estimation methods available: the two ``standard'' approaches are nonparametric maximum likelihood and nonparametric Bayes. The former maximizes the likelihood based on observations $X_1, \dots, X_n$ from $m_P$, with respect to $P$. Given the nonparametric nature of $P$, the resulting estimate is almost surely discrete and the points of support are no greater than $n$ \citep{lindsay1983}. The latter approach assigns a prior distribution to $P$, typically a Dirichlet process  \citep{ghoshramamoorthi2003, ghosal2017, hjortetal2010, ferguson1974}, and evaluate the corresponding posterior mean, given $(X_1,\ldots,X_n)$.  Even though there is no direct imposition of discreteness in the posterior, draws from the posterior distribution of $P$ have atoms \citep[e.g.,][]{blackwellmacqueen1973} and the corresponding posterior mean is spiky, ``effectively discrete.'' Hence, neither the likelihood nor Bayesian approaches give satisfactory solutions to the problem of estimating a mixing distribution $P$ in \eqref{eq:mix}.  The point is that, in these traditional approaches, the focus is on identifying candidate $P$ such that resulting mixture density $m_P$ is compatible with the empirical distribution of data, not specifically estimating the mixing distribution. 

A third general approach is available, which is the primary focus of this paper, called {\em predictive recursion} (PR).  Unlike the previous two methods, which are likelihood-based, the PR estimator is based on a stochastic, recursive algorithm that aims specifically to estimating the mixing distribution $P$ based on data from the mixture model \eqref{eq:mix}.  This strategy was first proposed in \citet{newtonetal1998} as a fast and smooth approximation to the posterior mean of $P$ under a Dirichlet process mixture model; see \citet{martin.pr.survey}. The idea behind PR is to start with a initial guess, $P_0$, and then update that guess recursively based on each individual observation $X_i$ for $i=1,\ldots,n$.  PR has a number of desirable computational and statistical properties.  First, PR is computationally efficient---its complexity is $O(n)$.  Second, the PR estimator, $P_n$, is absolutely continuous with respect to $P_0$, so if $P_0$ has a smooth density, then so does $P_n$.  Third, the PR estimator has also been shown to consistently estimate the true mixing distribution $P$ in a series of papers: \citet{tmg2009}, \citet{martintokdar2009}, and \citet{dixitmartin2021}.  

Applications of the PR algorithm have appeared in \citet{newton2002}, \citet{martintokdar2011}, \citet{martintokdar2012}, \citet{martinhan2016}, \citet{tansey2018}, \citet{woodyetal2022}, and  \citet{dixitmartin2022}. In each of these applications, however, the mixing distribution support $\UU$ is a relatively low-dimensional space, e.g., one- or two-dimensional.  The reason for this constraint is that, while the algorithm itself is completely general, computation of the normalizing constant in Equation \eqref{eq:PR} below can be a challenge when $\UU$ is more than two- or three-dimensional.  In particular, the required integration can only be done numerically, but efficient quadrature methods are available only when the domain of integration, in this case $\UU$, is low-dimensional.  A Monte Carlo-based strategy would be less sensitive to the dimension of $\UU$ and, in that sense, would have an advantage.  Unfortunately, no such Monte Carlo-based strategy is currently available in the literature, and this paper aims to fill this gap.  

Following a brief review in Section~\ref{S:background} of the PR algorithm and importance sampling techniques, we propose in Section~\ref{S:prticle} below the {\em PRticle filter} approximation. As the name suggests, this consists of an augmentation of the original PR algorithm with a filtering step that adaptively reweights an initial set of particles along the PR updating sequence.  The idea is that, at the $i^\text{th}$ step, the weighted set of particles resembles a sample from the PR estimate $P_{i-1}$ based on data $X_1,\ldots,X_{i-1}$. Hence, the $n^\text{th}$ step gives a particle approximation of the PR estimate $P_n$ and Theorem~\ref{thm:limit} below establishes that, for fixed data $X_1,\ldots,X_n$, this approximation converges almost surely in total variation distance to $P_n$ as the number of sampled particles approaches infinity.  

In Section~\ref{S:numerical}, we evaluate performance of the proposed PRticle filter approximation on both real and simulated data sets. For the simulated data sets, the evaluation is split into two types. First, to judge the accuracy of the proposed PRticle filter approximation, we compare it to the original PR estimator in cases where a quadrature scheme is feasible.  In our comparisons, the PRticle filter accurately approximates the PR estimate for simulations from mixtures corresponding to univariate and bivariate mixing distributions. Second, when the dimension of the mixing distribution support is too large for a quadrature scheme to be practical, we compare our PRticle filter approximation to a Dirichlet process mixture model-based estimator. The PRticle filter approximation is faster to compute and of comparable quality compared to the nonparametric Bayes estimator, which is one of the best known solutions.  

For a real data illustration, we consider an application where data consists of a marked spatial point process.  That is, the observed data consists of spatial locations at which specific events take place, along with some other relevant feature of the events, called marks.  As is common in spatial point process models, the relevant quantity is the intensity function.  Here we follow \citet{taddykottas2012} and model this intensity function as a mixture, with a multivariate mixing distribution support, and apply the PRticle filter approximation to estimate the mixing distribution and, in turn, the intensity function.  This naturally leads to estimates of other relevant features, including conditional distribution of the marks given the spatial locations.  We argue that the results obtained through our use of the PRticle approximation are consistent with patterns seen in the data and with those presented elsewhere in the literature.  This application simply would have been impossible using the basic PR algorithm.   
%{\color{red} PR is shown to give the required nonparametric flexibility to the joint intensity function of location and marks, which is possible due to the computational ease offered by the PRticle filter approximation.} 
Some concluding remarks are given in Section~\ref{S:discuss} and the proof of Theorem~\ref{thm:limit} is presented in Appendix~\ref{proofs}.

\section{Background}
\label{S:background}

\subsection{Predictive recursion}
% The PR algorithm is a fast recursive algorithm proposed in \citet{newtonetal1998} by Michael Newton with applications and theory appearing in \citet{newtonzhang1999}, \citet{newton2002}. It was proposed as an approximation to the posterior mean of the mixing distribution $P$ under a Dirichlet process prior. It is also a particular case of the Robbins-Monroe algorithm. 
Suppose we have data $X_1, \ldots, X_n$ from $m_{P}$ in \eqref{eq:mix}, where the goal is estimation of the mixing distribution $P$. With a user-defined initial guess $P_0$ and weight sequence $\{w_i: i \geq 1\} \subset (0,1)$, the $i^\text{th}$ step in the PR algorithm is given by,
\begin{equation}\label{eq:PR}
P_i(du) = (1-w_i) \, P_{i-1}(du) + w_i \, \frac{k(X_i \mid u) P_{i-1}(du)}{\int k(X_i \mid u) \, P_{i-1}(du)}, \quad i = 1, \ldots, n 
\end{equation}
For theoretical reasons, the weights must satisfy $\sum_{i=1}^\infty w_i = \infty$ and $\sum_{i=1}^\infty w_i^2 < \infty$; this can be achieved by taking, e.g., $w_i = (i+1)^{-\gamma}$ for some $\gamma \in (0.5, 1]$. The algorithm processes the $n$ data points sequentially and returns the final update $P_n$ as the PR estimator of the mixing distribution. The corresponding PR mixture density estimate is $m_n = m_{P_n}$, where the mapping $P \mapsto m_P$ is given in \eqref{eq:mix}. It is clear that the PR estimator $P_n$ depends on the ordering of the observations $X_1, \dots, X_n$. If this dependence is undesirable, then it can be removed---or at least mitigated---by calculating $P_n$ over multiple permutations of the data and averaging over the estimates \citep{newton2002, tmg2009}. With the superior computational efficiency of PR, this permutation-averaging can still be carried in a fraction of the run-time of its competitors. 

Key features of the PR algorithm/estimator include its ability to estimate a mixing density and its computational efficiency. By the former, we mean that if the user-defined initial guess $P_0$ has a smooth density with respect to a particular dominating measure, then the final PR estimator $P_n$ will too.  Compare this to the maximum likelihood and Bayes estimators, which are necessarily (or ``effectively'') discrete.  By the latter computational efficiency claim, we mean that each PR step requires a fixed number of computations, so the overall computational complexity of PR algorithm is $O(n)$. 

As mentioned in Section~\ref{intro}, the key step in each iteration of the PR algorithm is calculation of the normalizing constant $\int k(x_i \mid u) \, P_{i-1}(du)$. Since $P_{i - 1}$ is data-driven and fully nonparametric, we cannot expect there to be a closed-form expression for the normalizing constant. Often it can be approximated numerically using a quadrature scheme; this is especially easy to do so when the mixing distribution support $\UU$ is univariate. However, for as the dimension of $\UU$ increases, computation of the normalizing constant becomes more and more challenging.  For example, the number of grid points required for accurate quadrature grows exponentially in the dimension of $\UU$ and becomes infeasible or at least inefficient even for moderate $\text{dim}(\UU)$. This creates a computational bottleneck.  

In previous work, this challenge was side-stepped by treating some of the latent variables as mixing variables and the others as non-mixing/structural parameters.  For example, instead of mixing the kernel $k(x \mid u_1, u_2)$ over both the location $u_1$ and scale $u_2$, the proposal in \citet{martintokdar2011} was to treat, say, the scale parameter $u_2$ as a fixed unknown, so that mixing is required only over the univariate $u_1$-space.  Then they developed a PR-based marginal likelihood for $u_2$ that could be used for simultaneous estimation of the scale $u_2$ and the corresponding mixing distribution over $u_1$.  This effectively reduces the dimension of the mixing distribution support, thus making it easy to side-step the challenges in computing the normalizing constant.  For various reasons, however, it would be preferable to deal with the computational challenges directly, as opposed to using a ``hack'' to reduce the dimension artificially.  This requires new ideas for evaluating the normalizing constant in \eqref{eq:PR} and, for this, here we develop a novel strategy based on the same ideas behind sequential importance sampling.

\subsection{Importance sampling and filtering}
%\subsection{Importance sampling}

The approximation we propose in Section~\ref{S:prticle} uses the principles behind importance sampling and particle filters in general. Before stating our algorithm, we first review these basic principles. Consider the general problem of integrating a function $h$ with respect to a probability density $p$, where $U \in \UU \subset \RR^d$, for $d \geq 1$. 
%When $d=1$ and the form of $p$ is known, one can easily use a quadrature scheme on $N$ points to perform the calculation. For $d \geq 2$, the numerical calculation can still be carried out if the probability density $p$ is a smooth function with no significant highs or lows. This is because such a numerical integration can still be performed with a minimal number of grid points. On the other hand, if $p$ has a complex structure then a finer grid is required on $\RR^d$ which can exponentially increase the computation time. 
In cases where numerical integration is infeasible, e.g., if $d$ is too large or if either $h$ or $p$ is too rough, it is common to use a Monte Carlo approximation by averaging over a random set of observations from probability density $p$. However, a problem arises if $p$ cannot be efficiently sampled from. In such cases, an {\em importance sampling} approach can be employed. This amounts to generating samples from a different distribution, say with density $q$, and then reweighting those samples so that they resemble samples from $p$.  In particular, the expected value of $h$ with respect to $p$ can be written as
\[ \int_\UU h(u) \, p(u) \, du = \int_\UU h(u) \, \frac{p(u)}{q(u)} \, q(u) \, du, \]
and this immediately suggests the Monte Carlo approximation 
\[ \frac1T \sum_{t=1}^T \alpha_t \, h(U_t), \]
where $\{U_t: t=1,\ldots,T\}$ are iid samples from $q$ and $\alpha_t = p(U_t)/q(U_t)$ are the weight adjustment factors.  If the normalizing constant for $p$ is unknown, then the $T^{-1}$ factor can be replaced by $(\sum_{t=1}^T \alpha_t)^{-1}$. 

The ratio $\alpha_t = p(U_t) / q(U_t)$ helps to effectively filter out points in low $p$-density regions while increasing the weight put on particles in high $p$-density regions. \citet{agapiouetal2017} unify the existing literature on importance sampling with a special focus on determining the size of $T$ such that error in approximation is minimized.
%The discussion is mainly in terms of Bayesian inverse problems. 
The choice of $T$, the Monte Carlo sample size, is important, both in terms of accuracy and efficiency. A practical measure of efficiency used for importance sampling is the {\em effective sample size} (ESS), i.e., the effective number of particles. 
%This is inversely proportional to the variance of $\hat{\mu}$. 
Following \citet{kong1992}, a commonly used expression for ESS is 
\begin{equation}\label{eq:ess}
    \text{ESS} = \frac{ (\sum_{t=1}^T \alpha_t)^2}{\sum_{t=1}^{T} \alpha_t^2},
\end{equation}
where $\alpha_t = p(U_t)/q(U_t)$ as before.  By  Cauchy--Schwartz, ESS is bounded above by $T$, and the closer it is to $T$ the more efficient the importance sampler.  So the goal is to choose the proposal density $q$ such that ESS is as close to $T$ as possible.

These basics behind importance sampling can be connected to more sophisticated Monte Carlo methods with the following interpretation. The procedure above essentially starts with a collection of tentative sample points from $p$, which are commonly referred to as {\em particles}. Particles which have small importance ratios, $p/q$, are given small weight, and effectively {\em filtered out}. In this sense, importance sampling is a (basic) form of {\em particle filtering}. This idea can then be extended in different directions. In particular, it would be possible for the target distribution, $p$, to be changing over some ``time'' index. In hidden Markov models, for example, the dimension of the target distribution's support is increasing with time; also, in Bayesian inference, the target $p$ is the posterior distribution which is evolving with the sample size $n$.  Sequential Monte Carlo methods have proved useful in these problems \citep[e.g.,][]{doucetetal2001, doucetjohansen2011, deletal2006}. 
%Another instance is when $p$ is a Bayesian posterior distribution based on a sample of size $n$; then, of course, it would be changing with $n$.
% Alternatively, if one believes that $p$ is spiky, with well-separated modes, then one might consider a tempered version of the target, i.e., $p^{(s)}(u) \propto p^s(u)$, for $s \in (0,1]$.  For small $s$, the density would be less spiky and, thereby, easier to learn; so one can start with a proposal for $p^s(u)$ with a small $s$ and repeat the importance sampling procedure by gradually increasing $s$, hence updating the target towards $p(u)$. This way in the $s \to 1$ and $T \to \infty$ limit, accurate estimation of expected values with respect to $p$ could be achieved.(like in the Bayesian posterior) or for a better representation of the target (like in the tempered case).  
%This need to sequentially update the proposal distribution to reflect the changing target need not be specific to a Bayesian setup and can be a tool to better represent a target sequentially \citep[e.g.,][]{deletal2006}. 
Sequential importance sampling, in particular, is a powerful tool for particle filtering \citep{agapiouetal2017, tokdarkass2010}. In the context of mixture models, sequential importance sampling \citep[eg.][]{maceachernetal1999} and particle learning algorithms \citep[eg.][]{carvalhoetal2010particle} have been suggested for analyzing mixture models in the Bayesian setting. In our present case, sequential updating is required because we need particles that represent the PR estimate $P_i$ as $i=1,2,\ldots,n$. This problem is due to the unique recursive structure inherent in the PR sequence of target distributions and, therefore, calls for different or at least more specialized techniques compared to what is currently available in the sequential Monte Carlo literature.

\section{PRticle filter approximation}
\label{S:prticle}

\subsection{Algorithm}

In this section we propose a particle filter algorithm designed specifically to approximation the PR estimator. For simplicity, and without any real loss of generality, assume that $P_0$ has a density with respect to Lebesgue measure on $\UU \subset \RR^d$, denote by $p_0$.  Then all the subsequent PR updates $P_i$ have such a density, denoted by $p_i$.  At each iteration of PR, one needs to calculate a normalizing constant
\[ m_{i-1} (X_i) = \int_\UU k(X_i \mid u) \, p_{i-1} (u) \, du, \quad  i = 1, \ldots, n. \]
The analytical form of $p_{i-1}$ is unknown so clearly we cannot evaluate this in closed form.  Likewise, we cannot directly generate observations from it to get a Monte Carlo approximation.  However, we know that it is a function of the previous updates $p_0, \ldots, p_{i-2}$, so the idea is to leverage the PR algorithm's recursive formulation and those core importance sampling principles to design an efficient Monte Carlo/particle filter approximation. 

Recall that $p_0$ is a user-specified density on $\UU$ and we will assume that sampling from $p_0$ is feasible. Generate an iid sample $U_1, \ldots, U_T$ of size $T \gg 1$ from $p_0$. Then, a simple Monte Carlo average gives us an approximation of the first normalizing constant, 
\[ \hat m_0 (X_1) = \frac{1}{T} \sum \limits_{t=1}^{T} k(X_1 \mid U_t) \]
where each point $U_t$ is equally weighted by $T^{-1}$. Next, we do not know the form of $p_1$ but we know that it can be expressed in terms of $p_0$ and the data point $X_1$ as
\begin{align*}
p_{1}(u) & = (1 - w_1)p_{0}(u) + w_{1} \frac{k(X_1 \mid u) p_{0}(u)}{m_0 (X_1)} \\
& = \Bigl\{ 1 + w_1 \Bigl( \frac{k(X_1 \mid u)}{m_0 (X_1)} - 1 \Bigr) \Bigr\} p_0(u). 
\end{align*}
This implies the ratio of consecutive PR density estimates is 
\[ \frac{p_1(u)}{p_0(u)} = \delta_0(u) := \Bigl\{ 1 + w_1 \Bigl( \frac{k(X_1 \mid u)}{m_0 (X_1)} - 1 \Bigr) \Bigr\}. \]
Now, since 
\[ m_1(X_2) = \int k(X_2 \mid u) \, p_1(u) \, du = \int k(X_2 \mid u) \, \delta_0(u) \, p_0(u) \, du, \]
we have a very natural Monte Carlo approximation of $m_1(X_2)$, namely, 
\[ \hat m_1(X_2) = \frac{1}{T}\sum \limits_{t=1}^{T} k(X_2 \mid U_t) \, \hat\delta_{0}(U_t), \]
where $\hat\delta_0(u)$ is based on plugging in $\hat m_0(X_1)$ for $m_0(X_1)$ in the definition of $\delta_0$ above.  Here $\hat\delta_0(\cdot)$ acts as a mesh that effectively filters out those particles that are not compatible with the updated distribution $p_1$.  Continuing with the same logic, for the $i^\text{th}$ iteration, we get
\[ \hat m_{i-1}(X_i) = \frac{1}{T}\sum \limits_{t=1}^{T} k(X_i \mid U_t) \, \hat \Delta_{i}(U_t), \quad i \geq 1, \]
where $\hat \Delta_1(u) \equiv 1$ and 
\begin{align*}
\hat \Delta_{i}(u) & = \hat \Delta_{i-1}(u) \, \hat\delta_{i-2}(u) \\
%=\prod\limits_{j=2}^{i} \delta_{j-2} (u;X_{j-1}) = 
& = \prod\limits_{j=2}^{i} \left \{ 1 + w_{j-1} \left ( \frac{k(X_{j-1} \mid u)}{\hat m_{j-2} (X_{j-1})} - 1 \right ) \right \}, \quad i \geq 2. 
\end{align*}
The above steps make up the PRticle filter approximation and these are summarized in Algorithm~\ref{algo:prticle}.  In the end, the algorithm returns the pairs $\{(U_t, \hat \Delta_n(U_t)): t=1,\ldots,T\}$ that collectively represent an approximate sample from the PR estimator $P_n$.  From this sample, any features of $P_n$ can be approximated as usual.  If an estimate of the density $p_n$ were required, then the weighted collection of particles can be smoothed using, e.g., a kernel density estimator. Just like the PR estimator, the PRticle filter approximation depends on the ordering of observations, and same permutation-averaging can be used here to mitigate the order-dependence, if desired.

% \begin{alg}[{\color{red}PRticle filtering algorithm}]
% For iid data $X_1, \dots, X_n$ with initial guess $P_0$, corresponding random sample $U_1, \ldots, U_T$ and weight sequence $\{w_i: i \geq 1\} \subset (0,1)$, $i^\text{th}$ update in the PRticle filtering algorithm is given by,
% \[ P_i(du) = (1-w_i) \, P_{i-1}(du) + w_i \, \frac{k(X_i \mid u) P_{i-1}(du)}{\hat m_{i-1}(X_i)}, \quad i = 1, \ldots, n\]
% where
% \[ \hat m_{i-1}(X_i) = \frac{1}{T}\sum \limits_{t=1}^{T} k(X_i \mid U_t) \Delta_{i}(U_t) \]
% and
% \begin{equation*}
%   \Delta_{i}(U_t) = \left \{
%   \begin{aligned}
%     &1, && i = 1\\
%     &\prod\limits_{j=2}^{i} \left \{ 1 + w_{j-1} \left ( \frac{k(X_{j-1} \mid u)}{\hat m_{j-2} (X_{j-1})} - 1 \right ) \right \}, && i \geq 2
%   \end{aligned} \right.
% \end{equation*} 
% \end{alg}

\begin{algorithm}[t]
\SetAlgoLined
Initialize: Data $X_1, \ldots, X_n$, initial guess $p_0$, random sample $U_1, \ldots, U_T$ from $p_0$, and weight sequence $\{w_i: i \geq 1\} \subset (0,1)$\; 
Set $\hat \Delta_t = 1$ for $t=1,\ldots,T$\;
 \For{$i=1,\ldots,n$}{
  set $N_{t,i} = k(X_i \mid U_t) \, p_{i-1}(U_t)$ for each $t$, and $D_i = T^{-1}\sum_{t=1}^{T} k(X_i \mid U_t) \, \hat \Delta_t$\;
  update  $p_i(U_t) = (1 - w_i) p_{i-1}(U_t) + w_i \, N_{t,i} / D_i$ for each $t$\;
  evaluate $\hat \Delta_t = \hat \Delta_t [1 + w_i \{k(X_i \mid U_t) / D_i - 1\} ]$ for each $t$\;
  %$D(u_t) = D(u_t)*(1 + w(i)*(k(X_i \mid U_t)/\text{den} - 1))$\;
  }
 return $U_t$ and weights $\hat \Delta_t$, for $t=1,\ldots,T$.
 \caption{\textbf{PRticle filter approximation}}
 \label{algo:prticle}
\end{algorithm}

\subsection{Convergence}

A relevant question would be of the convergence of the PRticle filter approximation $\hat p_n = \hat p_{n,T}$ to the corresponding PR estimate $p_n$, as the Monte Carlo sample size $T$ goes to $\infty$. If we had just one data point $X_1$, then $\hat p_{1,T} \to p_1$ simply by the law of large numbers as this only involves the simple Monte Carlo approximation of $m_0(X_1)$. However, as we include more observations, the $i^\text{th}$ approximation of $p_i$ consists of the previous $i-1$ approximations and the law of large numbers argument is not immediately clear. But it turns out that the law of large numbers can be applied to show that PRticle filter approximation, $\hat p_{n,T}$, converges to its target $p_n$ in a very strong sense as $T \to \infty$. 

%Regardless of this, since the principle behind the PRticle approximation is of a sophisticated Monte Carlo procedure, we use an induction based argument to prove that for a fixed data set $X_1, \dots, X_n$, $\hat p_n \to p_n$ as $T \to \infty$. In spite of this, with the knowledge of recursiveness of PR and the initial Monte Carlo approximation, we prove that for a fixed set of observations $X_1, \ldots, X_n$ the PRticle filter approximation does achieve the PR estimate in the limiting case.

\begin{thm}
\label{thm:limit}
For a fixed data set $X_1,\ldots,X_n$, let $p_n$ and $\hat p_{n,T}$ denote the PR estimator and its PRticle filter approximation, respectively, both based on the same initial guess with distribution $P_0$.  If the kernel is such that 
\begin{equation}
\label{eq:kernel.assumption}
\int_\UU \Bigl\{ \prod_{i \in {\cal S}} k(X_i \mid u) \Bigr\} \, P_0(du) < \infty, \quad \text{for all ${\cal S} \subseteq \{1,\ldots,n\}$}, 
\end{equation}
%$u \mapsto k(x \mid u)$ is bounded for each $x$, 
then 
\[ \int_\UU |\hat p_{n,T}(u) - p_n(u)| \, du \to 0, \quad \text{with $P_0$-probability~1 as $T \to \infty$.} \]
% Given observations $X_1, \ldots, X_n$ from $m$, if $m$ is modeled by $m_P$ with form \eqref{eq:mix} and the corresponding mixing density $p$ is estimated by the PR procedure in \eqref{eq:PR}, then the PRticle filter approximation $\hat p_n \equiv \{(U_t, \Delta_n(U_t)): t=1,\ldots,T\}$ based on $T$ particles is consistent for the PR estimate $p_n$. Mathematically,  
% \[\lim \limits_{T \to \infty} \hat p_{n} \to p_{n}\]
\end{thm}

\begin{proof}
See Appendix~\ref{proofs}. 
\end{proof}

Theorem~\ref{thm:limit} establishes that, with a sufficiently large Monte Carlo sample size $T$, the PRticle filter approximation, $\hat p_{n,T}$, of the PR mixing density estimator $p_n$ will be quite accurate.  Note that $L_1$/total variation convergence implies weak convergence, so virtually any relevant functional of $p_n$ can be accurately approximated by the corresponding functional of $\hat p_{n,T}$.  The condition \eqref{eq:kernel.assumption} on the kernel is rather mild, e.g., it is satisfied if $u \mapsto k(x \mid u)$ is bounded for almost all $x$. Beyond the fixed-data approximation, the result in Theorem~\ref{thm:limit}, together with the general results in \citet{martintokdar2009} and \citet{dixitmartin2021} on the consistency properties of $p_n$ as $n \to \infty$, suggests that $\hat p_{n,T}$ would also be a good estimator of $p$ when both $n$ and $T$ are large.

\subsection{Adaptation to handle attrition}
\label{attrition}

The final estimate $p_n$ will depend on the initial $p_0$, not just through the default PR mechanism but also through the dependence on the choice of particles $U_1,\ldots,U_T$ from $p_0$.  To ensure that $p_n$ captures the true mixing density $p$, it is generally recommended to choose a relatively diffuse $p_0$ in the PR algorithm. 
%This ties in with the recommendation for a proposal density in the importance sampling procedure. 
However, the true $p$ is likely to be more concentrated in certain regions of $\UU$ than in others.  So those chosen particles $U_1,\ldots,U_T$ from $p_0$ that happen to fall in those $p$-low-density regions of $\UU$ should be assigned relatively low weights.  The concern is that too many of the particles end up in these low-density regions, hence affecting the effective number of particles. Recall that, an efficiency measure of the particle filter is given by the effective sample size (ESS) in \eqref{eq:ess}. For our case this can be calculated as,
\[\text{ESS} = \frac{\bigl\{\sum_{t=1}^{T} \Delta_n(U_t)\bigr\}^{2}}{\sum_{t=1}^{T} \Delta_n(U_t)^{2}}. \]
If too many particles end up with negligible weights, i.e., if $\Delta_n(U_t) \approx 0$ for $t$', then ESS becomes significantly smaller than $T$. This loss-of-information, called {\em attrition}, is a common problem in importance sampling or particle filtering  \citep[eg.][]{doucetjohansen2011}; and it cannot be ignored because the effective sample size is what controls the accuracy of the Monte Carlo approximations. To account for this, the general strategy is to resample points from the region of importance such that ESS does not reduce tragically \citep[e.g.,][]{carvalhoetal2010particle, doucetetal2001}. The strategy we propose here is in the same spirit as adaptive importance sampling \citep[eg.][]{bugallo2017}. Below we describe our approach that accounts for attrition, specific to the PRticle filter.  We start by summarizing the characteristics of the final PR estimate $p_n$ to improve upon the initial filter $U_t$. This summary can then be used to construct a new informed $p_0$ so that an updated filter has more points in the more dense regions of $p$.

Given points $U_1, \ldots, U_T$ and the final weights $\Delta_{n}(U_t)$ representing the PR estimate $p_n$ we can easily obtain Monte Carlo approximations of 
\[\mu_n = \int u \, p_n(u) \, du \quad \text{and} \quad \Sigma_n = \int (u - \mu_n) (u - \mu_n)^\top \, p_n(u) \, du, \]
the mean vector and covariance matrix associated with $p_n$, respectively, given by 
\[ \hat\mu_n = \frac1T \sum\limits_{t=1}^{T} U_t~ \Delta_n (U_t) \quad \text{and} \quad \widehat\Sigma_n = \frac1T \sum\limits_{t=1}^{T} (U_t - \hat\mu_n) (U_t - \hat\mu_n)^\top \Delta_n (U_t). \]
% an estimate of the mean $\E(U)$ and variance $\V(U)$ as below,
% \[\hat \E(U) = {\color{red} \frac1T} \sum\limits_{t=1}^{T} U_t \Delta_n (U_t)  \]
% \[ \left [\hat \V(U) \right ]_{ij} = \sum\limits_{t=1}^{T} U_{ti}U_{tj} \Delta_n (U_t) - [\hat \E(U)]_{i} [\hat \E(U)]_{j}, \quad i = 1,\ldots, d , j = 1, \ldots, d \]
This helps us to identify a region where $p_n$---and likely $p$ as well--- has high concentration. This information can then be incorporated in a new updated $p_0$. A reasonable strategy, therefore, is to redefine the initial estimate $p_0$ to be, e.g., a multivariate Student-t distribution with location $\hat\mu_n$ and scale matrix $\widehat\Sigma_n$. An iid sample is then generated from this new $p_0$ and the PRticle procedure is carried out as before. 

This idea of updating the PRticle filter can be extended in various  ways. One way is to repeat the aforementioned process more than once. However, in our experience, this can lead to shrinkage of the region of interest beyond of what is needed resulting into points from only a highly dense region and no points elsewhere.  Alternatively, one could identify several particles having relatively large weight following the initial pass of PRticle filter; use these as locations around which a multivariate Gaussian or Student-t distribution could be centered; and then take the updated $p_0$ to be a mixture of these few distributions and sample particles from there. Yet another strategy is to resample particles from $p_0$ after one run of PR and rerun the algorithm by replacing the low-weight particles by the new particles. This allows for identifying new regions of interest while removing any low-probability regions. This is a classical strategy of resampling used in particle filters for Bayesian problems \citep{carvalhosmoothing2010}. Any of the approaches suggested above would be useful in reducing attrition of particles, but more deliberation is needed to conclude which of these strategies would be most efficient. For our purposes here, we use the strategy mentioned before and our simulation results in Section \ref{S:numerical} show that this is effective in reducing attrition.

%Another way of ensuring a better representation of the mixing distribution and its various modes of high dense regions is to first pick a top percentage of points in $U_t$ that correspond to large weights in $p_n$. Then we identify any clusters that exist in these set of points and use a new guess $f_0$ as a mixture of densities (say normal) over each cluster. 

\section{Numerical results}
\label{S:numerical}

%We present numerical results for the PRticle filter approximation in two parts. First we show comparisons of the PRticle filter approximation to PR, in univariate and bivariate simulated scenarios, followed by a multivariate implementation. 

\subsection{Density estimation}

Here we show three density estimation examples.  Examples~1--2, involving Euclidean data and data on a sphere, respectively, compare the PRticle filter approximation to the original PR estimator in low-dimensional cases where the latter can be computed efficiently.  Example~3 considers cases where the mixing distribution support is too high-dimensional to compute the original PR estimator, so we compare the PRticle filter approximation results to those of the Dirichlet process mixture model fit. 

\begin{ex}
For $d$-dimensional data $X$, consider a normal mixture model of the form \eqref{eq:mix} with $k(x \mid u) = \nm_d(x \mid u, \sigma^2 I_d)$ the multivariate normal density with mean vector $u$, where $I_d$ is the $d$-dimensional identity matrix.  Throughout, $\sigma^2=0.5$ will be taken as fixed.  So that we can compare the PRticle filter approximation directly to the original PR estimator, we consider only the cases $d=1$ and $d=2$ here.  For the $d=1$ case, we take the true mixing distribution to be $P= \bet_{[0,10]}(10,5)$, a beta distribution scaled to $\UU = [0, 10]$; for the $d=2$ case, we take $P = \bet_{[0,10]}(10,5) \times \bet_{[0,10]}(5,10)$, a joint distribution supported on $\UU=[0,10]^2$ corresponding to independent scaled beta marginals.  In both cases, samples of size $n=500$ are generated and we compare the PRticle filter approximation to the original PR estimator in terms of the {\em Kullback--Leibler divergence} $K(m_n, \hat m_n)$, where $m_n$ is the PR estimator of the mixture density and $\hat m_n$ is the corresponding PRticle filter approximation.  Both are based on weight sequence $w_i = (i+1)^{-1}$ and initial guess $P_0 = \unif(\UU)$.  The PRticle filter approximation relies on samples $U_1,\ldots,U_T$ taken from $P_0$ and here we consider four samples sizes, $T \in \{100, 300, 500, 1000\}$.  Table~\ref{tab:firstcomparisons} summarizes both the Kullback--Leibler divergence and the ESS for both the $d=1$ and $d=2$ cases.  As expected, the ESS tends to be smaller for $d=2$ than for $d=1$, with the former retaining about 15\% of the original sample while the latter retains about 33\%. However, the Kullback--Leibler divergence tends to be small across the board and does not vary much as a function of $T$ for both cases. 
\end{ex}

\begin{table}
\centering
\begin{tabular}{ll}
\begin{tabular}{*{3}{c} }
\hline
    $T$ & ESS & $K(m_n, \hat m_n)$ \\
    \hline
    100 & 33.6 & 0.0063 \\
    300 & 105.3 & 0.0001 \\
    500 & 162.9 & 0.0003 \\
    1000 & 330.7 & 0.0002\\
    \hline
\end{tabular}
%\caption{d=1}
& 
\begin{tabular}{*{3}{c} }
\hline
    $T$ & ESS & $K(m_n, \hat m_n)$ \\
    \hline
    100 & 10.5 & 0.0600 \\
    300 & 44.8 & 0.0326 \\
    500 & 81.7 & 0.0260 \\
    1000 & 140.7 & 0.0200\\
    \hline
\end{tabular}
%\caption{d=2}
\end{tabular}
\caption{Numerical results for Example~1: comparisons between $m_n$ and $\hat m_n$ with $d=1$ (left) and $d=2$ (right).  Comparisons are made in terms of effective sample size (ESS) and Kullback--Leibler divergence $K(m_n, \hat m_n)$.}
    \label{tab:firstcomparisons}
\end{table}

\begin{ex}
Next, following \citet{dixitmartin2022}, we compare the PRticle filter approximation to the original PR estimate for mixture models on the unit sphere $\SS \subset \RR^3$ commonly used for directional data.  The particular mixture model we consider is one with a so-called angular Gaussian distribution \citep{tyler1987} kernel 
\[k(x \mid \mu, \beta) \propto |\Sigma_{\mu, \beta}|^{-1/2} (x^\top \, \Sigma_{\mu, \beta}^{-1} \, x)^{-3/2}, \quad x \in \SS, \quad (\mu, \beta) \in \SS \times (0,\infty), \]
where $\Sigma_{\mu,\beta} = Q_\mu^\top D_\beta Q_\mu$, with  $D_\beta = \text{diag}(1,1,\beta^{-2})$ and $Q_u$ is the rotation matrix mapping $(0,0,1)^\top$ onto the unit vector $\mu \in \SS$, given by  
\[ Q_\mu = \begin{pmatrix} \cos{\theta_\mu}\cos{\phi_\mu} & -\sin{\phi_\mu} & \sin{\theta_\mu}\cos{\phi_\mu} \\ 
                        \cos{\theta_\mu}\sin{\phi_\mu} & \cos{\phi_\mu} & \sin{\theta_\mu}\sin{\phi_\mu} \\
                        -\sin{\theta_\mu} & 0 & \cos{\theta_\mu}\\ 
                        \end{pmatrix},
\]
and $(\theta_\mu, \phi_\mu)$ is the spherical coordinate representation of $\mu$.  
%{\color{red} This distribution is {\em bipolar axial} for $\beta < 1$, uniform over the sphere for $\beta = 1$, and {\em girdle} for $\beta > 1$ {\color{blue} (ref)}.}
For the original PR estimator, \citet{dixitmartin2022} treated $\beta$ as a fixed unknown structural parameter, not a latent variable being mixed over.  That is, they treated the kernel as $k_\beta(x \mid u)$, depending on the unknown $\beta$, where $u=\mu$ is the only latent variable mixed over.  Then they employed the PR marginal likelihood strategy to estimate the fixed unknown $\beta$.  Here, using the added flexibility of the PRticle filter approximation, we fit the model that mixes over latent variable $u=(\mu, \beta)$, so that there are no unknown structural parameters to be estimated separately.  Here we generate $n=2000$ samples from the above mixture model where the true mixing distribution $P$ has a smooth bimodal density in $\mu$ and a point mass at $\beta=0.1$---this means that PR's mixture model, that takes $\beta$ fixed and unknown, is correctly specified while the PRticle filter's mixture model is misspecified.  
% {\color{red} For the mixing distribution, we fix $\beta = 0.1$ and assume a bimodal continuous distribution on $\mu$, i.e,
% \[p(\theta_0, \phi_0) = 0.5 \, \mathsf{trN}_2((\pi/4, \pi/2)^\top, \Sigma) + 0.5 \,  \mathsf{trN}_2(\pi/4, 5\pi/4)^\top, \Sigma)\]
% \[\Sigma = \begin{pmatrix} (\pi/12)^2 & 0 \\ 0 & (\pi/6)^2 \end{pmatrix}\]
% where the support of the mixing distribution is expressed in terms of spherical coordinates such that $\UU = \{(0,\pi/2) \times (0,2\pi)\}$ and $\mathsf{trN}_2$ is a bivariate normal density truncated to lie on $\UU$. The idea here is that $(\theta_0, \phi_0)$ is a location parameter and $p(\theta_0, \phi_0)$ is induced on a sphere via a Jacobian transformation.}
For the PR estimator, we take $w_i = (i+1)^{-1}$ and $P_0$ to be uniform on $\SS$.  For the PRticle filter approximation, which mixes over both $\mu$ and $\beta$, the initial guess $P_0$ is a product of uniform distributions on $\SS$ and a uniform distribution on $(0,0.5]$.  
%For the traditional PR approach, $\beta$ is treated as a fixed structural parameter and estimated by a PR marginal likelihood approach \citep[see][]{dixitmartin2022} while the mixing distribution on $\UU$ is estimated via a quadrature scheme. With the PRticle filter approximation, we have the flexibility to define a joint mixing distribution over $\{\mu = (\theta_0, \phi_0), \beta\}$, where $(\theta_0, \phi_0) \in \UU$ and $\beta \in (0, 0.5)$. The algorithm is initialized with a uniform density over $\{(0,\pi/2) \times (0,2\pi) \times (0, 0.5)\}$, weight sequence $w_i = 1/(i+1)$ and particle size of $T = 1000$. 
Plots of the PR estimate $m_n$ and PRticle approximation $\hat m_n$, based on $T=1000$ initial particles, are provided in Figure \ref{fig:sphere.prticle}. The approximation based on PRticle filter clearly captures all the relevant features of the PR estimate, and in much less time thanks to not needing to employ the marginal likelihood strategy to estimate a fixed $\beta$. 
\end{ex}

\begin{figure}[t]
    \centering
    \subfloat[\centering PRticle estimate, north pole
    %on the positive Z-axis
    ]{
        \includegraphics[width=7 cm]{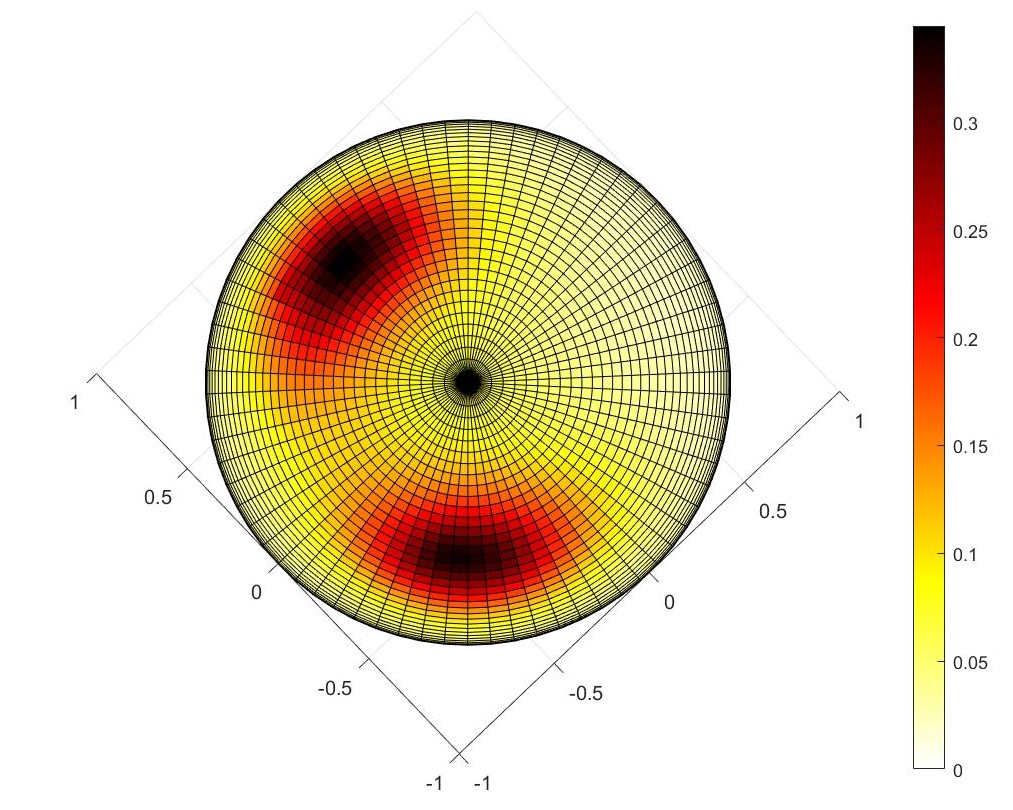}
     }
    \subfloat[\centering PRticle estimate, south pole
    %on the negative Z-axis
    ]{
        \includegraphics[width=7 cm]{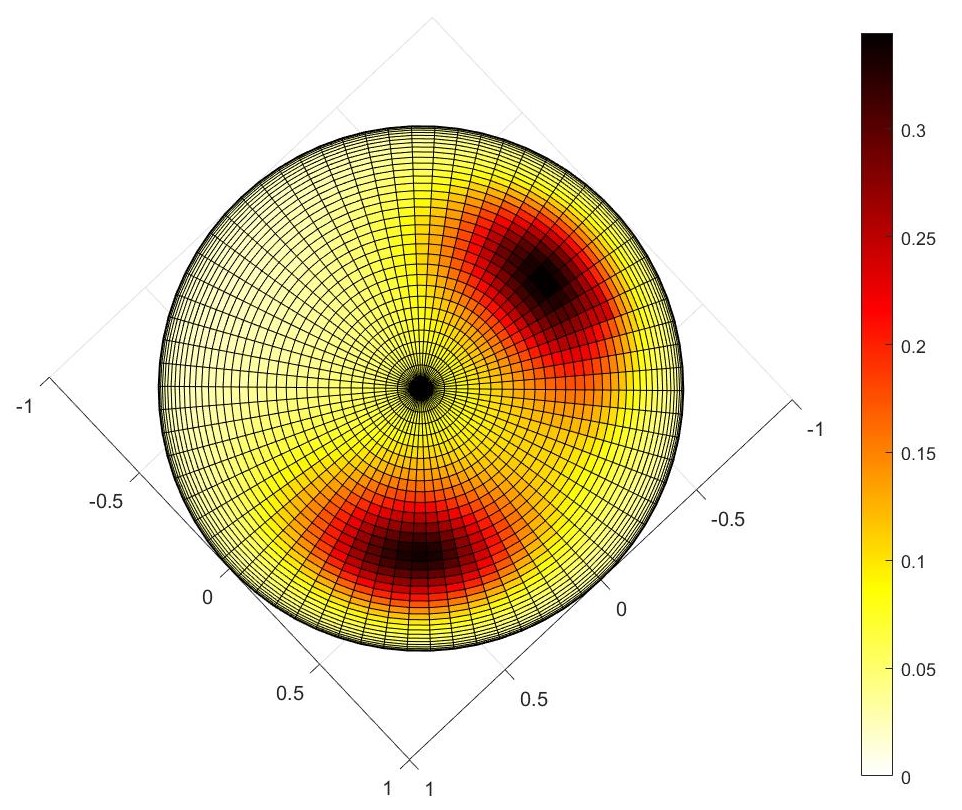}
     }
     \hspace{0mm}
    \subfloat[\centering PR estimate, north pole
    %on the positive Z-axis
    ]{
        \includegraphics[width=7 cm]{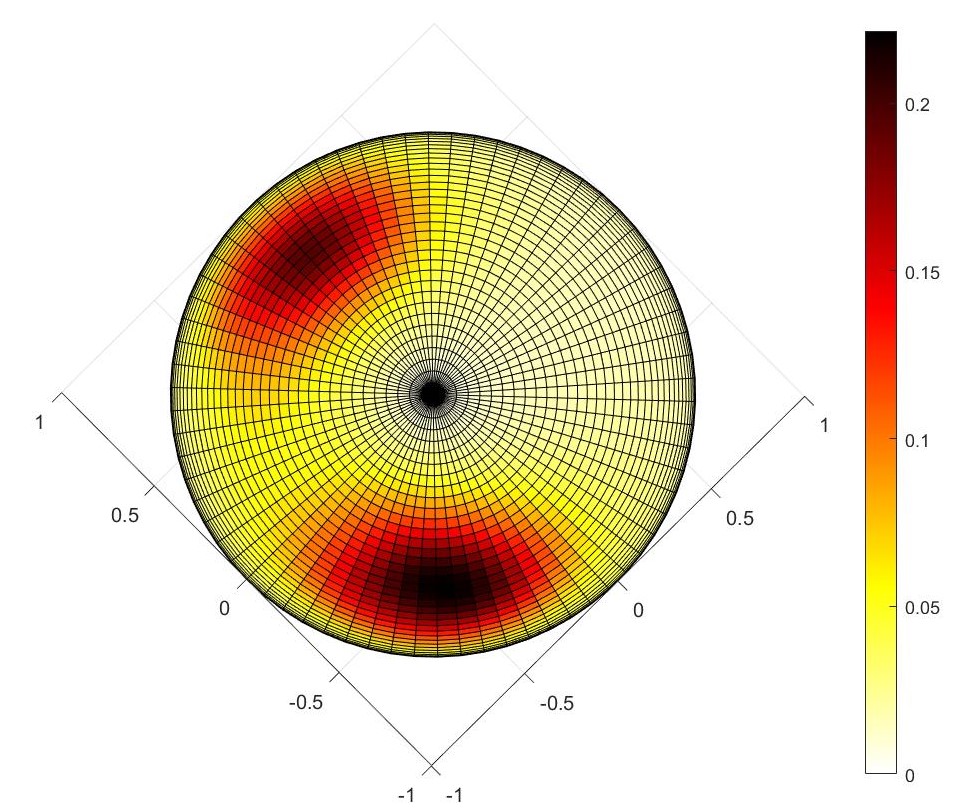}
     }
     \subfloat[\centering PR estimate, south pole
     %on the negative Z-axis
     ]{
        \includegraphics[width=7 cm]{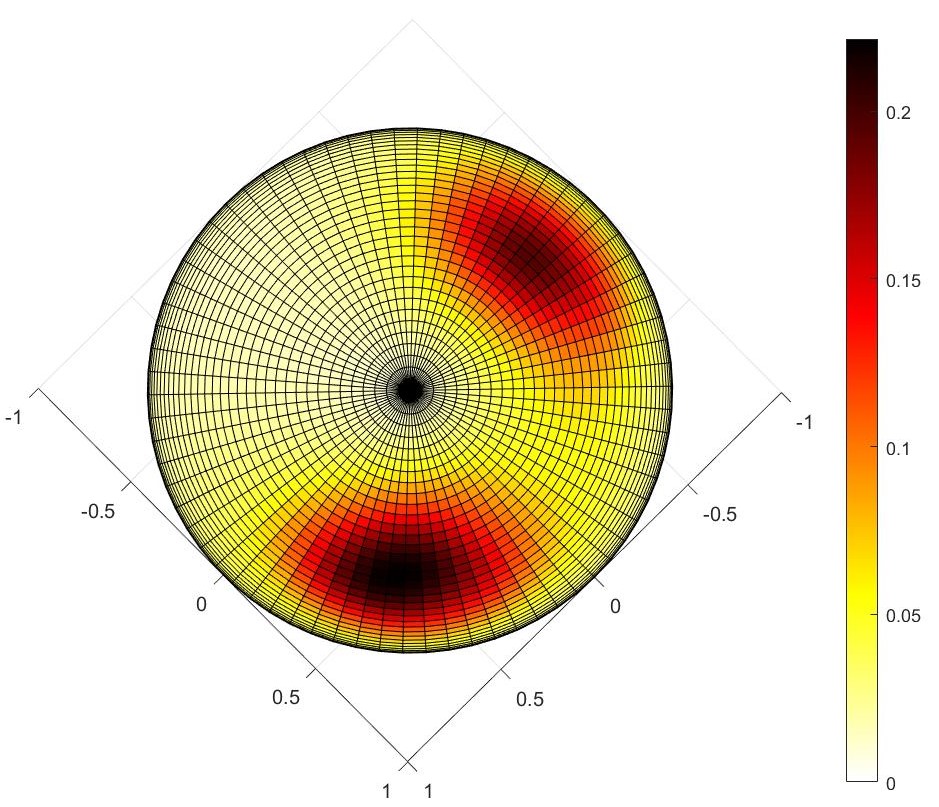}
     }
    \caption{Estimated mixture density on the sphere based on the PR algorithm and the PRticle filter approximation for bimodal continuous mixing distribution, views from north and south poles.}
    \label{fig:sphere.prticle}
\end{figure}

% \begin{equation}\label{eq:multvar_norm}
% m(x) = \int \nm(x \mid \mu, \Sigma) d P(\mu, \Sigma)
% \end{equation}

\begin{ex}
For the third part of the simulation study, we mix a bivariate normal kernel over all mean $(\mu_1, \mu_2)$ and covariance $(\sigma_1^2, \sigma_2^2, \rho)$ parameters. This means that the mixing distribution $P$ is defined over five variables. Using PR with numerical integration is not possible in this situation as a quadrature scheme is infeasible. The PRticle filter approximation can instead be used to fit this mixture density. For comparison, we consider a Dirichlet process mixture model fit, where the prior for the mixing distribution is $P \sim {\sf DP}(\alpha, P_0)$, a Dirichlet process with precision parameter $\alpha > 0$ and base measure $P_0$, which we take to be the same as PR's initialization (see below). The Dirichlet process mixture model estimate of the mixture density is the corresponding posterior mean, which is calculated using the {\tt DirichletProcessMvnormal} function in the R package {\tt dirichletprocess} \citep{rossmarkwick2018} with 1000 iterations. To compare the two approaches we take $U=(\mu_1, \mu_2, \sigma_1^2, \sigma_2^2, \rho)$ with the true mixing distribution $P$ corresponding to independent $\mu_1 \sim \nm(5, 3^2)$, $\mu_2 \sim \nm(10, 3^2)$, $\sigma_1^2 \sim \gam(1,1)$, $\sigma_2 \sim \gam(5,1)$, and $\rho \sim \bet(10, 5)$. 
% {\em Case 2:} Same as above except $\mu_1 \sim  0.3 \, \delta_{0} + 0.7 \, \delta_{15}$, $\mu_2 \sim 0.3 \, \delta_0 + 0.7 \, \delta_5$, and $\rho \sim \bet(5, 10)$, where $\delta_x$ denotes a Dirac point-mass distribution at $x$. 
% For the first case we use a normal density for the locations, gamma densities for the variance parameters and a beta density for the correlation parameter {\color{red} --- are you describing $P$?}. This results in a unimodal mixture density. For the second case, we use two-point probability densities for locations and the same gamma and beta densities for the variance-covariance parameters. This results into a bimodal mixture density.
% The key difference between these two cases is that Case~1 admits a unimodal mixture density where as Case~2 will result in a bimodal mixture density.
In this, we generate $n = 500$ observations from the true mixture density and fit a multivariate normal mixture density using the PRticle filter approximation and the Dirichlet process mixture model machinery. As before we initialize the PRticle filter with a uniform distribution $P_0$ over all parameters and a weight sequence $w_i = (i+1)^{-1}$. However, to avoid possible attrition we improve the filter by using the strategy proposed in Section~\ref{attrition} and rerun the algorithm with an updated $P_0$. Contour plots of the estimated mixture densities are given in Figure~\ref{fig:MVN1}. The PRticle filter approximation plots are able to capture the structure of the true mixture density, $m$, just like the Dirichlet process mixture model fit. For a numerical comparison we calculate the Monte Carlo approximation of the Kullback--Leibler divergence between the true mixture density and the estimated density. This is $0.024$ for a comparison between $m$ and $\hat m_{PR}$ while it is $0.006$ for a comparison between $m$ and $\hat m_{DP}$. The Dirichlet process estimate performs slightly better than PR for mixture density estimation, but it is important to note that PR is solving the harder problem of estimating a multivariate mixing density, which the Dirichlet process mixture formulation struggles with because the resulting estimator is effectively discrete. To illustrate this, we draw independent samples of $U$ from the true mixing distribution $P$ and both the PRticle filter and Bayes estimates of $P$, and display quantile--quantile plots for comparison in Figure~\ref{fig:MVNmixing}. The PR quantiles match the true distribution quantiles much more closely compared to the Dirichlet process-based Bayes estimator quantiles.  Computationally, fitting of the Dirichlet process mixture takes almost four minutes on our machine, while the PRticle filter approximation is calculated in about one minute.
\end{ex}

% \begin{table}[]
%     \centering
%     \begin{tabular}{|c|c|c|}
%     \hline
%     Case & $K(m, \hat m_{PR})$ & $K(m, \hat m_{DP})$\\
%     \hline
%     1 & 0.024 & 0.006\\
%     \hline
%     2 & 0.193 & 0.142\\
%     \hline
%     \end{tabular}
%     \caption{Numerical results for Example~3: Comparisons between DP and PR are made in terms of a Monte-Carlo approximation of the Kullback--Leibler divergence between the true density and estimate.}
%     \label{tab:KLmulti}
% \end{table}

\begin{figure}[t]
\centering
\subfloat[\centering DP mixture estimate]{
  \includegraphics[width=8 cm]{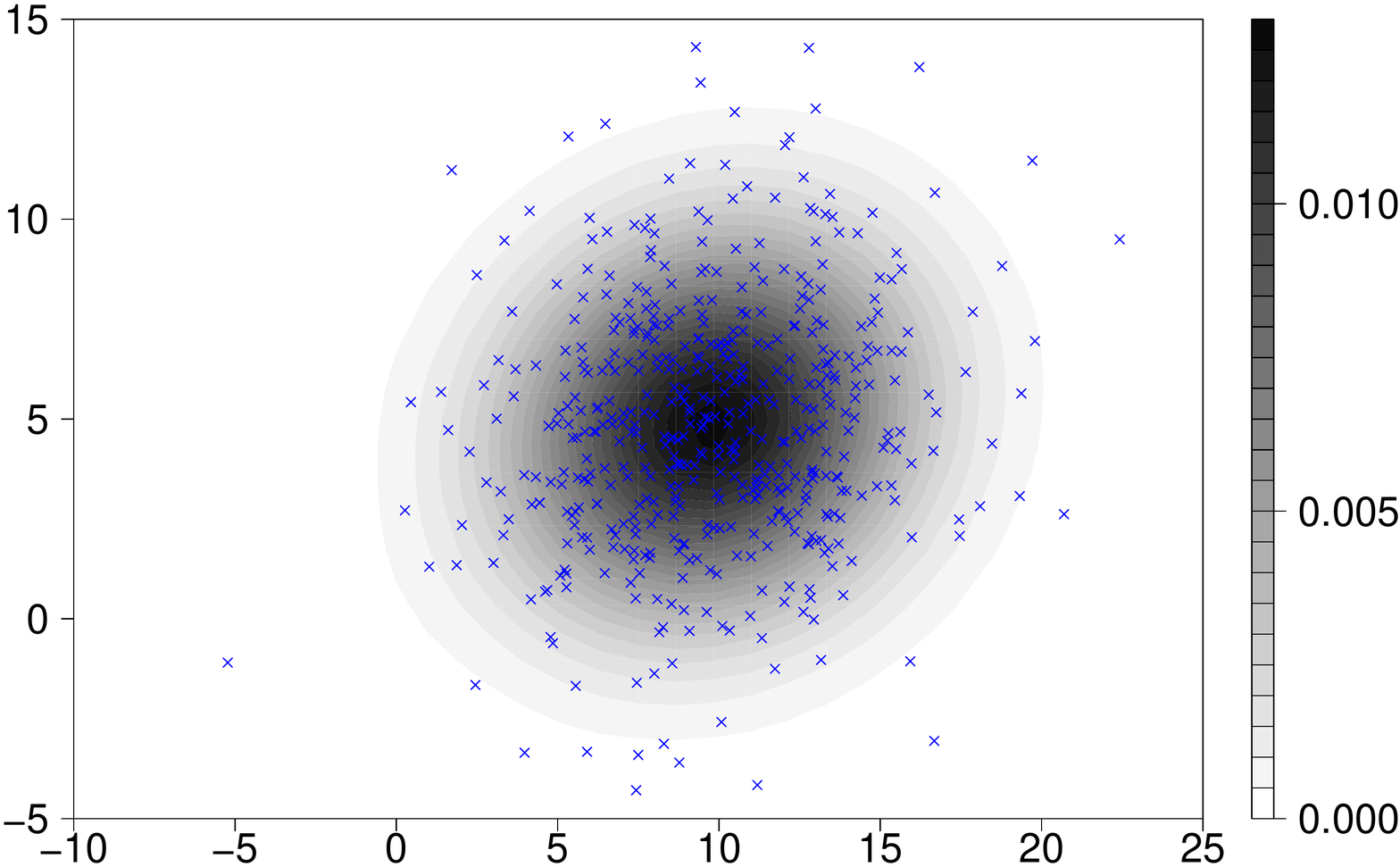}
}
\subfloat[\centering PRticle estimate with only one run ]{
  \includegraphics[width=8 cm]{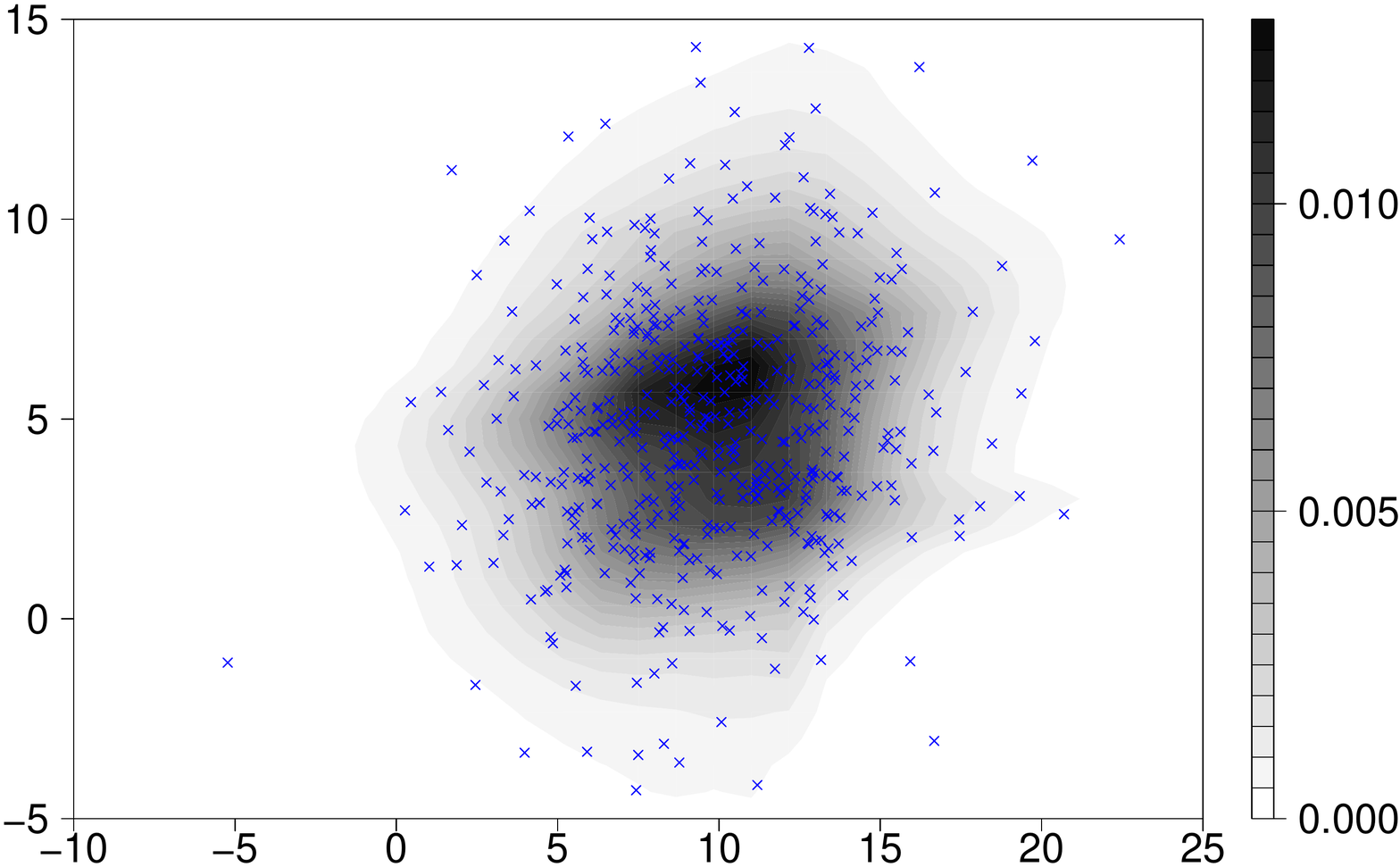}
}
\hspace{0mm}
\subfloat[\centering PRticle estimate with updated $p_0$]{
  \includegraphics[width=8 cm]{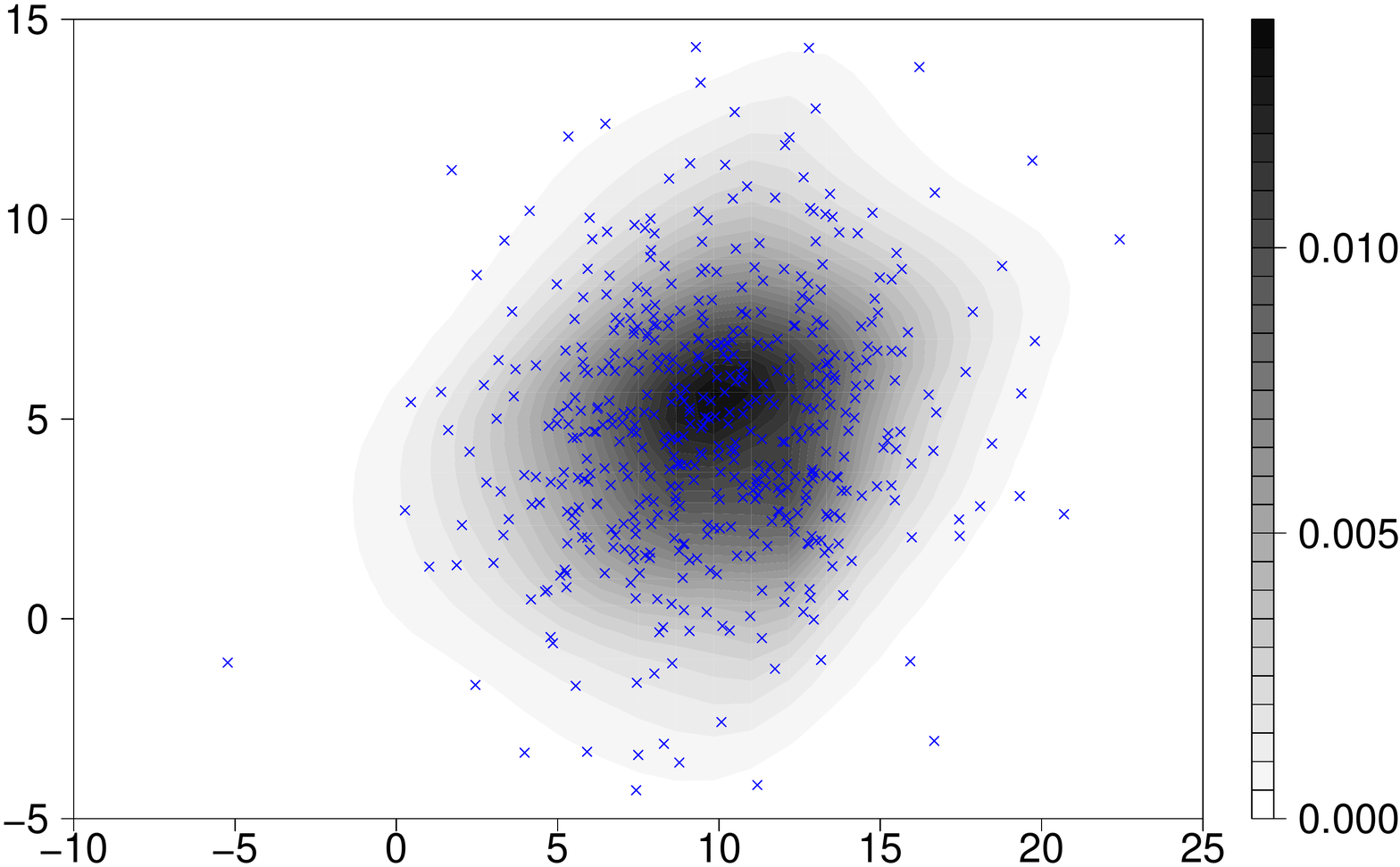}
}
\caption{Mixture density estimates for a multivariate normal mixture with the observed data overlaid.}
\label{fig:MVN1}
\end{figure}

\begin{figure}[t]
\centering
\subfloat[\centering $\mu_1$]{
  \includegraphics[width=5cm]{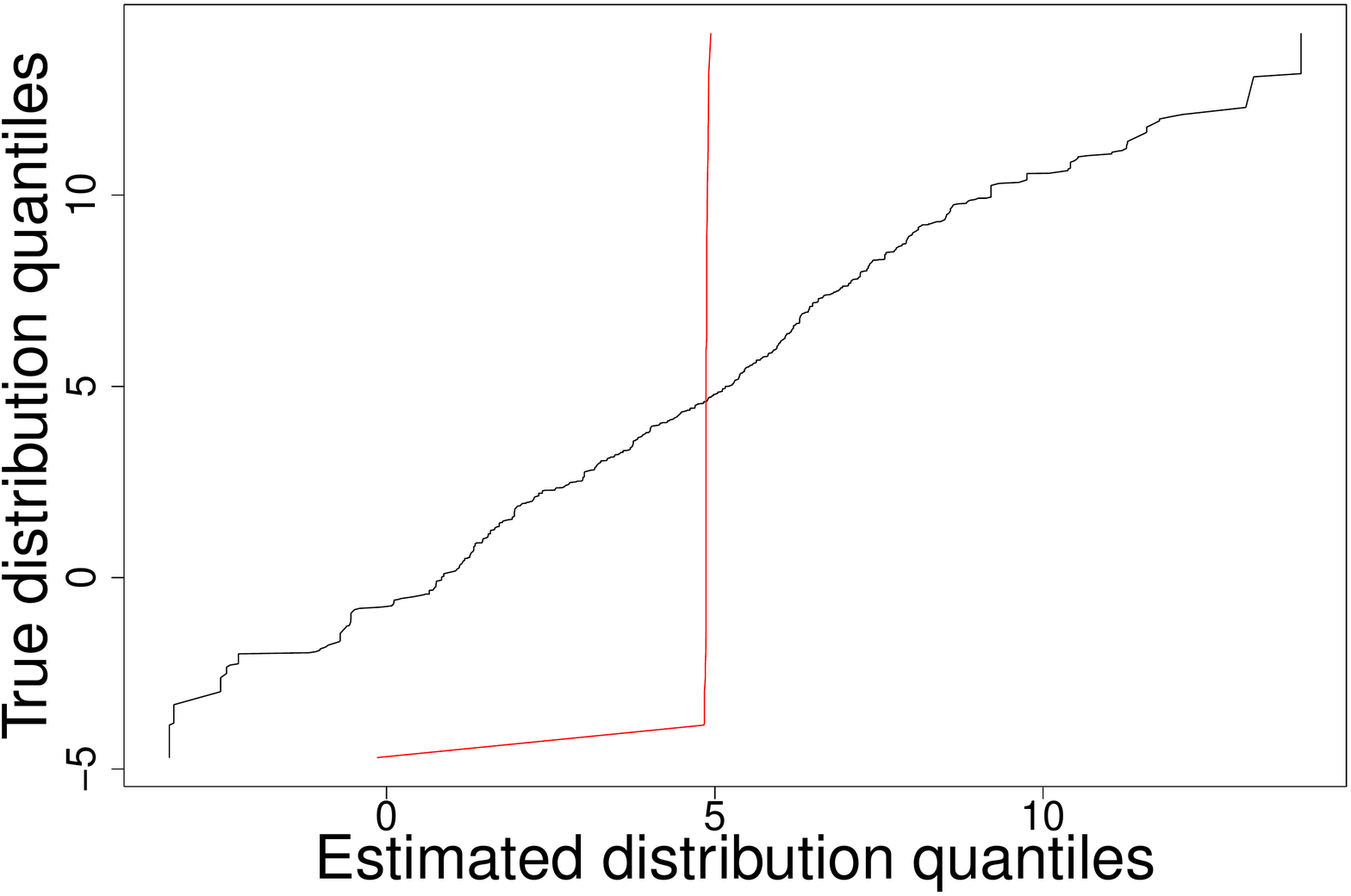}
}
\subfloat[\centering $\mu_2$]{
  \includegraphics[width=5cm]{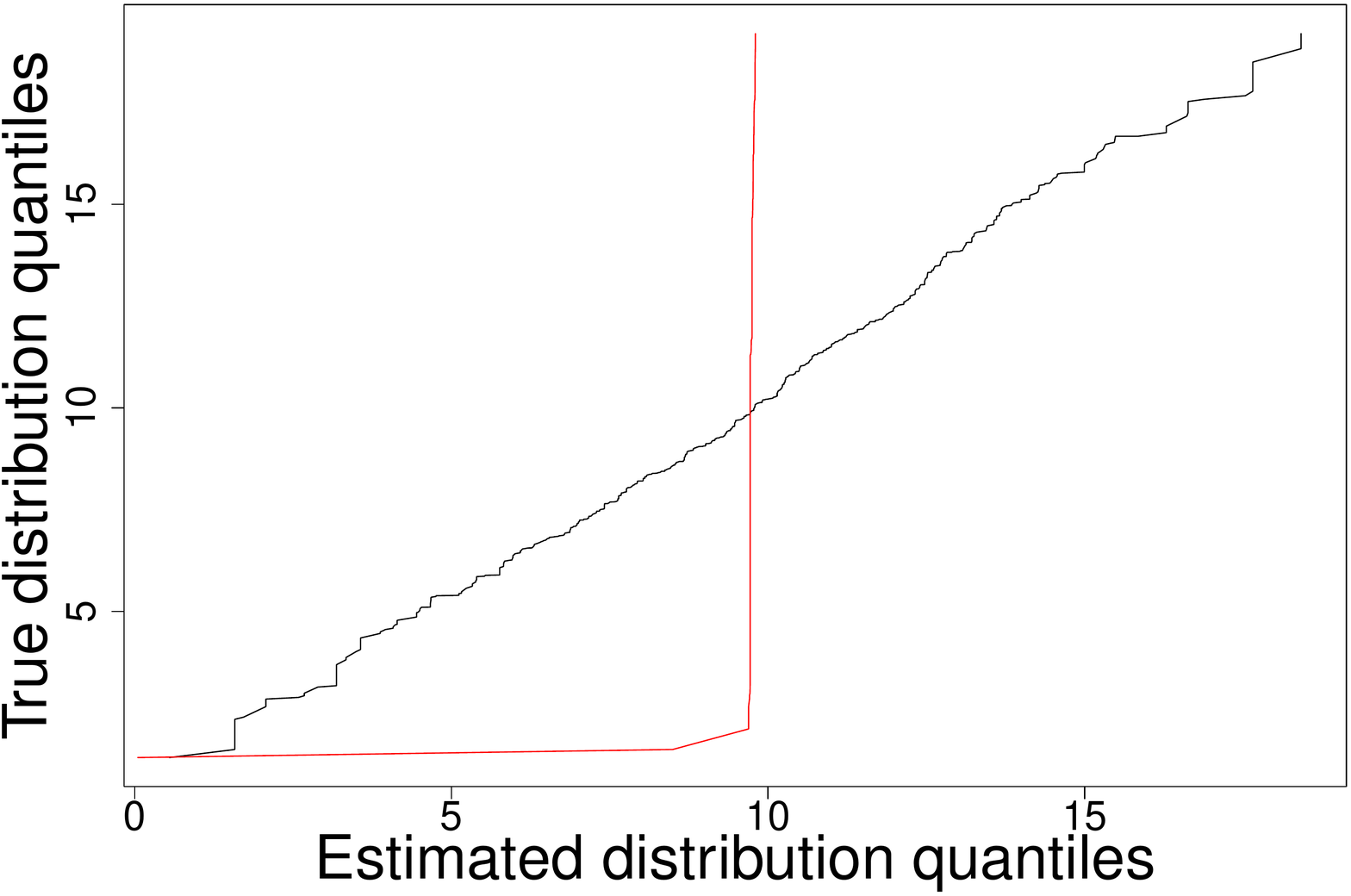}
}
\hspace{0mm}
\subfloat[\centering $\sigma_1^2$]{
  \includegraphics[width=5cm]{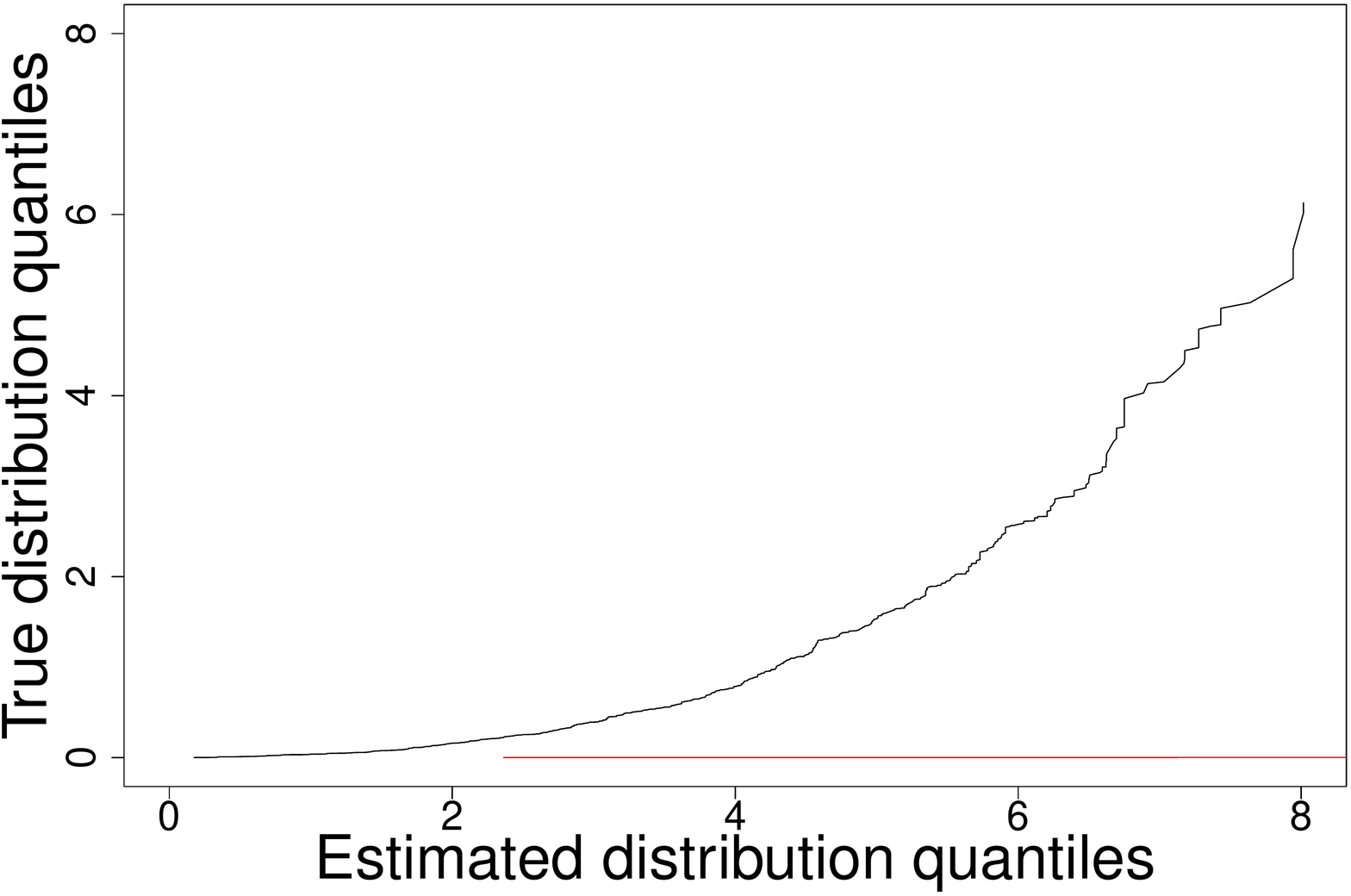}
}
\subfloat[\centering $\rho$]{
  \includegraphics[width=5cm]{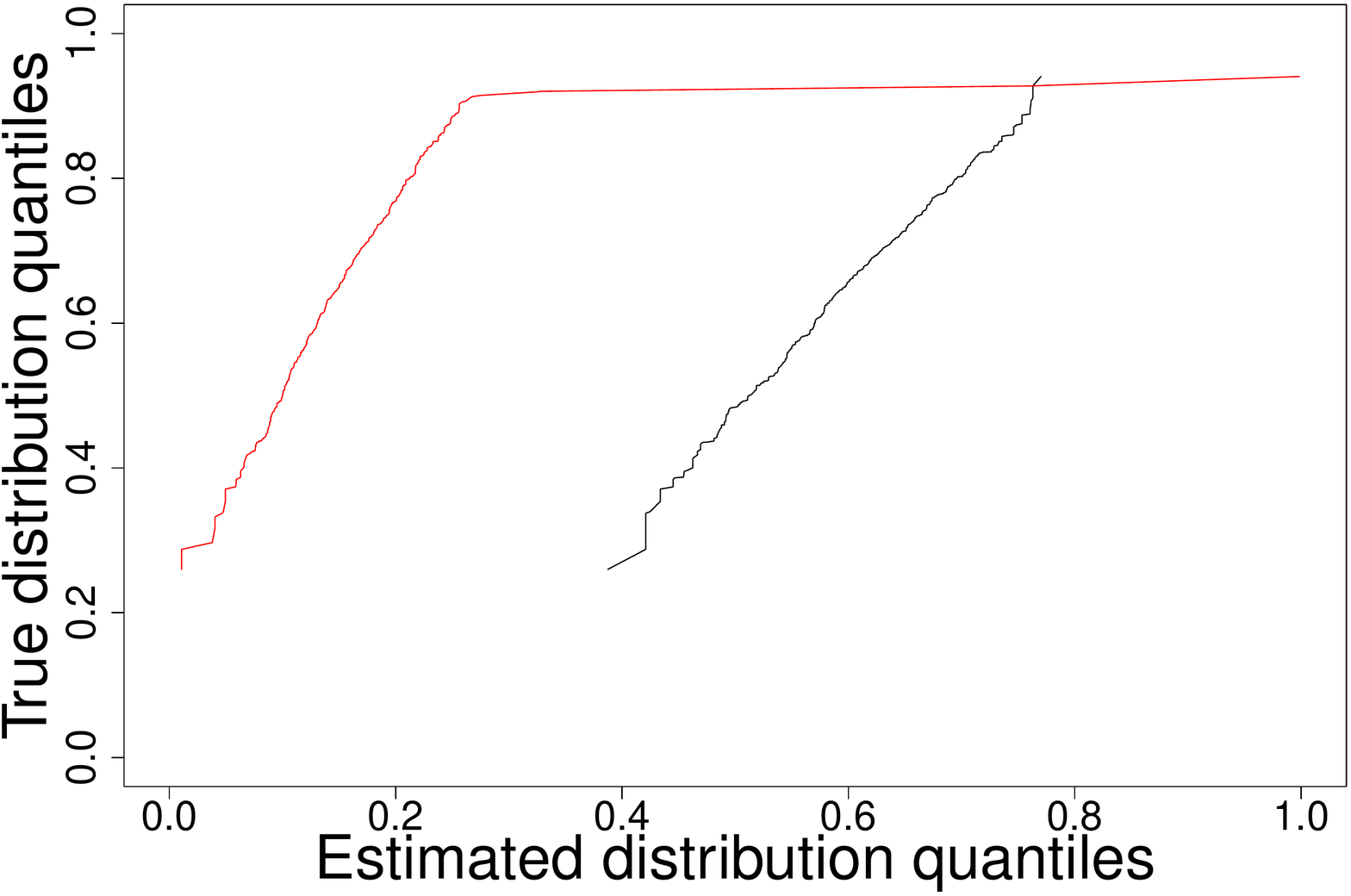}
}
\subfloat[\centering $\sigma_2^2$]{
  \includegraphics[width=5cm]{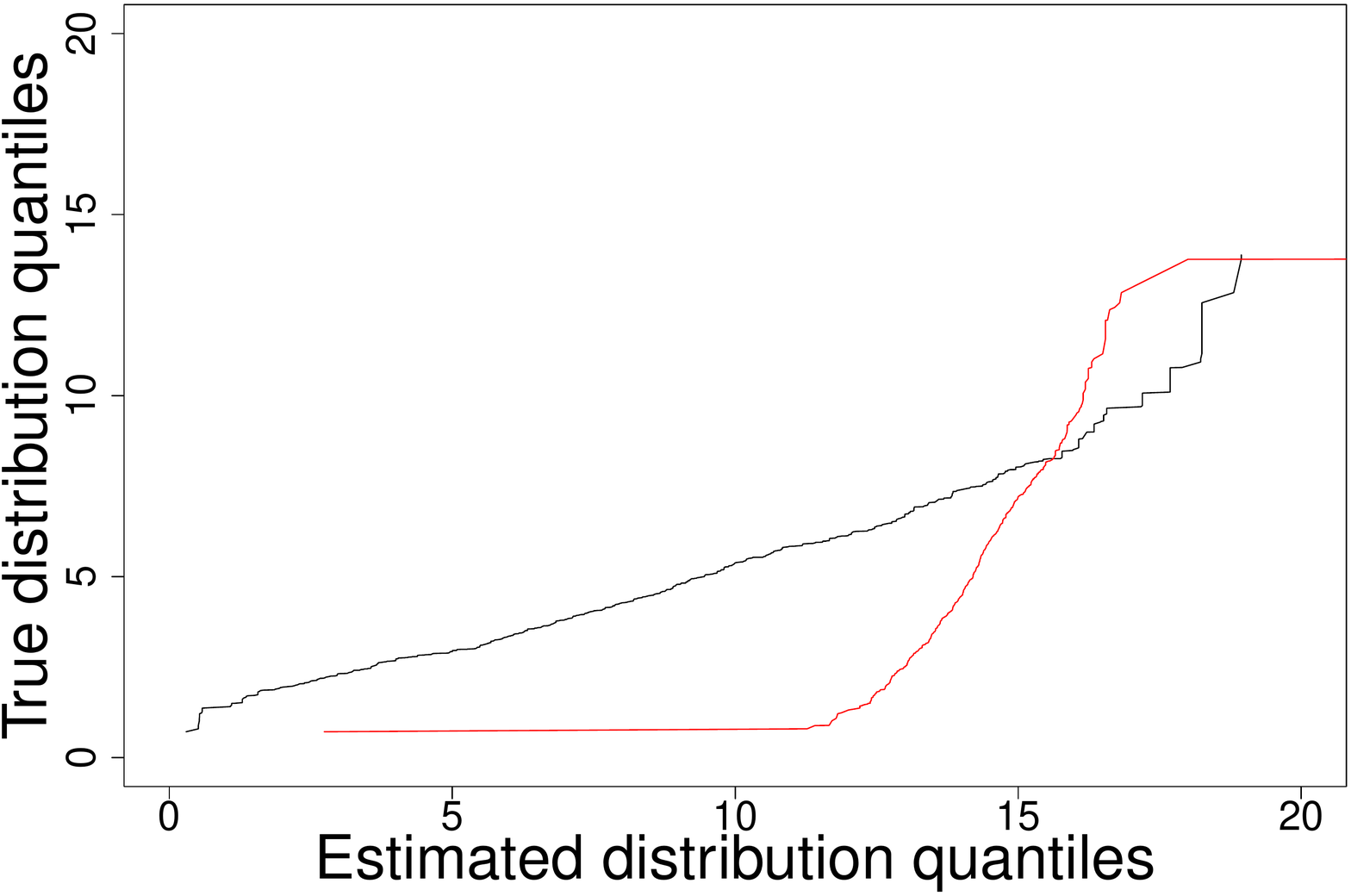}
}
\caption{Quantile--quantile plots for each component of $U=(\mu_1, \mu_2, \sigma_1^2, \sigma_2^2, \rho)$ of the mixing distribution corresponding to the multivariate normal mixture. Black line corresponds to quantiles from the PRticle filter estimate, while the red line corresponds to quantiles from the Dirichlet process-based Bayes estimate.}
\label{fig:MVNmixing}
\end{figure}

% \begin{figure}[t]
% \centering
% \subfloat[\centering DP mixture estimate]{
%   \includegraphics[width=7.5 cm]{figs/DPplotcase2.pdf}
% }
% \subfloat[\centering PRticle estimate with only one run]{
%   \includegraphics[width=7.5 cm]{figs/PRplotcase2_1.pdf}
% }
% \hspace{0mm}
% \subfloat[\centering PRticle estimate with updated $p_0$]{
%   \includegraphics[width=7.5 cm]{figs/PRplotcase2_22.pdf}
% }
% \caption{Mixture density estimates for Case~2 of multivariate normal mixtures with the observed data overlaid.}
% \label{fig:MVN2}
% \end{figure}

\subsection{Marked point process modeling}
\label{SS:marks}

Here we showcase an interesting application of multivariate mixture modeling using PR, which is made possible by the PRticle filter approximation. Suppose our data consists of spatial locations $s$ of an interesting occurrence possibly accompanied by some attributes $x$ at those locations. Typically, when only location observations $s_1, \ldots, s_n$ are available, there is interest in the intensity of the incident occurrence. These are typically modeled as realizations from a non-homogeneous Poisson process with intensity function $\lambda(s)$, $s \in \SS \subseteq \RR^d$ \citep[e.g.][]{liangetal2008}. For example, in an epidemiological study, $s_i$ might be the geographic location of the $i^\text{th}$ individual showing symptoms of a particular disease and hence there is interest in modeling the intensity of the disease occurrences. 
%In another case, $s_i$ could be the location on a basketball court from which a shot is made and the goal is to estimate the {\color{magenta} mean distribution of shots as a function of location?} 
For such a non-homogeneous Poisson process, a likelihood function can be written as,
%(see {\tt season2017} dataset in R). 
\[ L(\lambda \mid s_1,\ldots,s_n) = \Lambda^n \exp\{ -\Lambda \} \prod_{i=1}^{n} m(s_i)\]
where $m$ is the normalized intensity function, i.e., $m(s) = \lambda(s) / \Lambda$, and $\Lambda = \int \lambda(t) \, dt$. Given the separable nature of the likelihood above, $\Lambda$ and $m$ can be estimated separately. A regression approach is to model $\lambda$ by a log Gaussian Cox process \citep[ e.g.][]{liangetal2008}. However given the nonparametric nature of the problem it is desirable to use a robust model for $\lambda$ to capture all the shape/scale features of the function. Mixture models offer this flexibility and an approach to modeling $\lambda$ or $m$ by a Dirichlet process mixture was proposed in \citet{kottassanso2007}.

Additionally, there could be other attributes $X_1, \dots, X_n$ present with the location data, for example, indicator variable for type of disease, when there is interest in the association between disease locations. Then to account for this association and its effect on the model, a joint intensity function $\psi(s, x)$ can be defined. The resulting process is known as the marked point process, where the attributes are called {\em marks}. The nonparametric mixture density in \eqref{eq:mix} offers the required flexibility to model a fully nonparametric function $\psi(s,x)$. \citet{taddykottas2012} propose mixture models for such marked point processes using conditionally conjugate Dirichlet process mixture models. The idea is to model the joint intensity $\psi(s, x)$ of the locations $s$ and marks $x$ as,
\begin{equation}\label{eq:jointmark}
\psi(s, x) = \lambda(s) \, g(x \mid s) = \Lambda \, m(s) \, g(x \mid s) = \Lambda \, m(s, x),
\end{equation} 
where $g(x \mid s)$ represents the conditional density of mark $X$, given location $s$. Features of this joint intensity can be identified by modeling $m(s, x)$ with a mixture model. The flexibility and computational efficiency offered by PR means that it is tailor-made to
fit such a mixture model. However, given the multivariate nature of the problem we need the PRticle filter approximation to actually implement PR.

We illustrate the above on a real dataset as suggested in Example~5.3 of \citet{taddykottas2012}. The suggested dataset, {\tt longleaf} is part of the R package {\tt spatstat} \citep{bt2005} and a detailed space-time survival analysis based on this was developed in \citet{rc1994}. The observations are locations of 584 pine trees in a $200 \times 200$ square and the marks are diameters of the trees at breast height (only for trees having this diameter greater than 2 cm). A scatter plot of the data is given in Figure~\ref{fig:scatter}. One can clearly see that the distribution of trees is not uniform, i.e., mature (larger diameter) trees are more evenly distributed than younger (smaller diameter) trees, which appear in clusters. Hence, the goal is to model the joint intensity of the locations and marks of these trees. \citet{taddykottas2012} model $m(s, x)$ as a mixture model with a trivariate normal kernel and a mixing distribution defined over all the parameters of this multivariate normal distribution, i.e,
\begin{equation}\label{eq:marked}
    m(s, x) = \int \frac{\nm_3 \bigl(\text{logit} \, (s_1/200, s_2/200), \, \log (x-2) \mid \mu, \Sigma \bigr)}{(x-2)\prod_{i=1}^{2} (s_i / 200) (1 - s_i/200)} \, P(d\mu, d\Sigma)
\end{equation}
With the model in \eqref{eq:marked}, we can estimate the conditional distribution of the marks at different locations to capture the varying distribution of trees, which in essence is an indication of the survival. We propose using the PR approach to fit this joint intensity function and estimate $P$.  Of course, that this is a mixture of a nine-dimensional latent variable space---three mean parameters and six covariance matrix parameters---makes it impossible to fit with the PR algorithm directly, so the PRticle filter approximation is necessary. Assuming a mixing distribution over all nine dimensions is possible using the PRticle filter approximation, but for model comparison we actually fit two models: the nine-dimensional model above and a reduced model that assumes the covariance terms in $\Sigma$ are fixed at 0.
%one in which a mixing distribution is defined over all nine parameters and second where we simplify the kernel such that the covariance parameters are fixed to be zero and a mixing distribution is defined over the three location parameters and the three variance parameters. 
The mixing distribution is then estimated by PR with the PRticle filter approximation. From this fitted mixture model $m(s, x)$ we extract the conditional density $g(x \mid s)$ at specific locations to see how the diameter distribution varies with $s$ as displayed in Figure \ref{fig:marked}. As we can see in the scatter plot, each chosen location has unique characteristics in terms of diameter distribution. Locations $s = (81, 120)$ and $s = (100,100)$ have higher concentrations of mature, large-diameter trees, which is correctly captured by both models (nine-dimensional and six-dimensional) in Figure \ref{fig:marked}. On the other hand, locations $s = (105, 140)$ and $s = (185, 87)$ have clusters of younger, smaller-diameter trees which, again, is correctly captured in Figure~\ref{fig:marked}. 
Each plot in Figure \ref{fig:marked} is overlaid with an empirical probability density of marks using the {\tt density} function in {\tt R} based on observations that are within a radius of 30 units from the chosen location.
The fitted model retains these local features  while being globally smoother than the empirical density. In terms of model comparison, both the six and nine-dimensional model reasonably capture the varying diameter distribution at all locations. A difference between the two estimates is that the full model estimate is smoother than the reduced model one. This is because the kernel density in the former inherently contains an average over the covariance parameters, while the latter fixes these at zero.  The full model also appears to capture certain features better than the reduced model.  For example, consider the locations $s=(105,140)$ and $s=(185,87)$, whose conditional mark density is shown in Panels (c) and (d) of Figure~\ref{fig:marked}, respectively.  These two points have relatively high concentration of small-diameter trees, as seen in Figure~\ref{fig:scatter}; but upon closer inspection, the concentration at $s=(105, 140)$ seems higher than at $s = (185, 87)$, and we see that the conditional density estimates based on the full model capture these differing features better than those based on the reduced mode.  
%Hence, the full model is beneficial in the sense that it shows us the strongest feature of each location unlike the reduced model, which can tend to give unnecessary emphasis to more marks (e.g. $s = (81, 120)$ in Figure \ref{fig:marked}).
Similar results were obtained in \citet{taddykottas2012} via their proposed Dirichlet process mixture fit. An interesting difference between our results and those of Taddy and Kottas is that their plot at $s=(100,100)$ shows a sharp spike in the conditional density near $x=0$, whereas ours does not.  Since there is no evidence in the scatter plot for a high concentration of small-diameter trees, our guess is that their spike is actually a boundary effect, commonly seen in density estimation on bounded domains, and not an inherent feature in the data.  That the PR estimate does not suffer from a boundary effect in this case is another benefit.

\begin{figure}[t]
        \centering
        \includegraphics[width=13 cm]{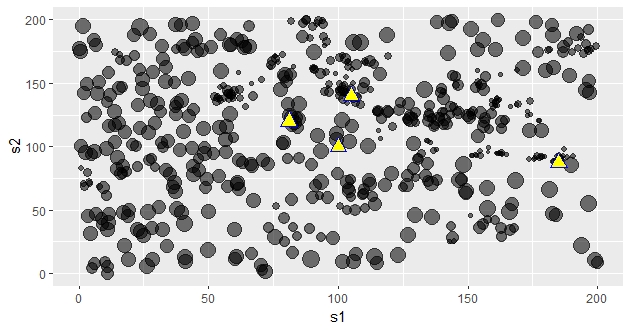}
        \caption{Tree locations on a $200 \times 200$ grid ({\tt longleaf} dataset), where the size of each point is proportional to the respective tree diameter; gray coloring is to make different points easier to distinguish. Yellow triangles indicate locations at which the conditional mark density is estimated in Figure \ref{fig:marked}.}
        \label{fig:scatter}
\end{figure}     

\begin{figure}[t]
    \centering
    \subfloat[\centering $s=(81, 120)$]{
    \includegraphics[width=7.5 cm]{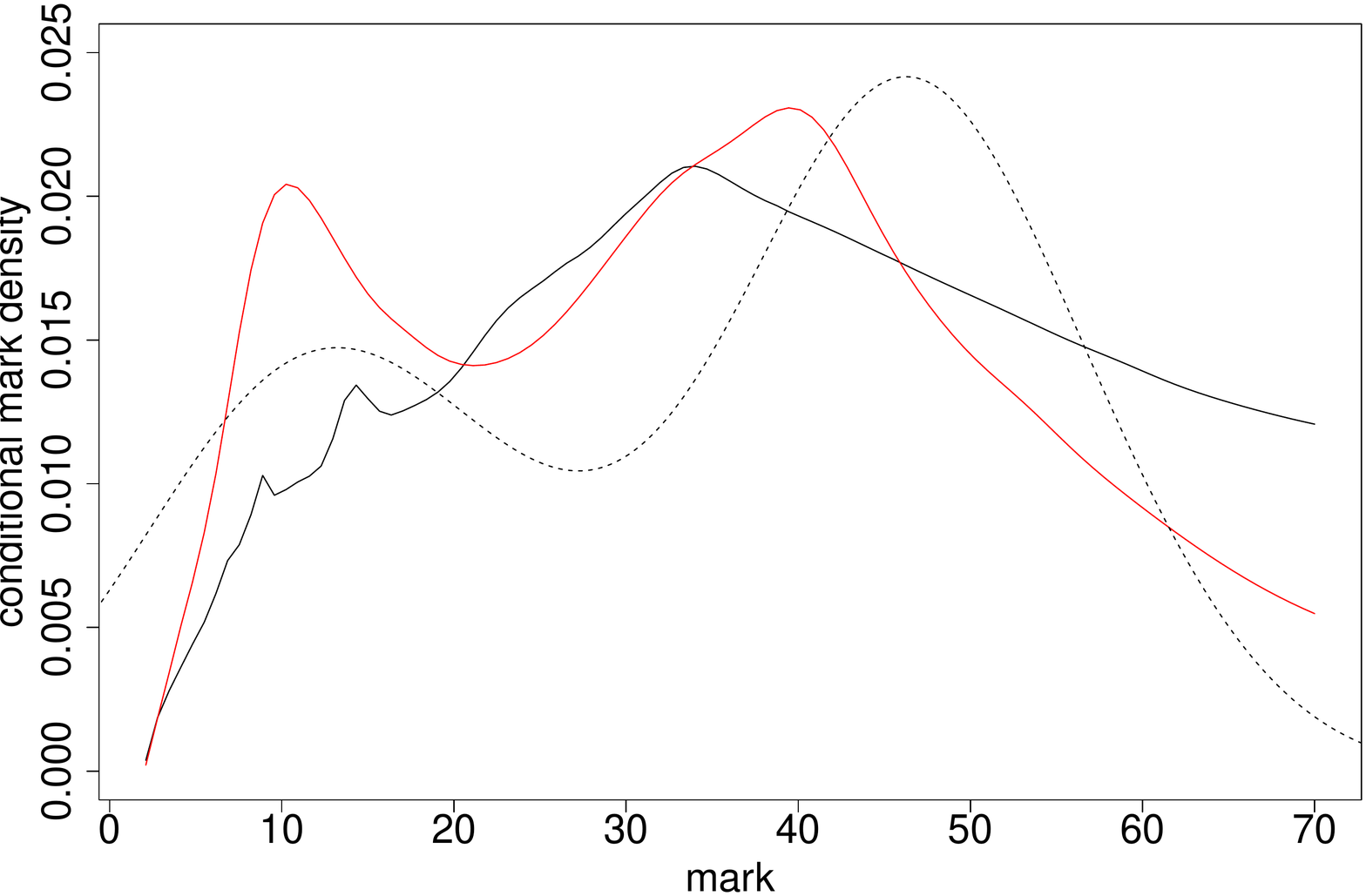}
    }
    \subfloat[\centering $s=(100, 100)$]{
    \includegraphics[width=7.5 cm]{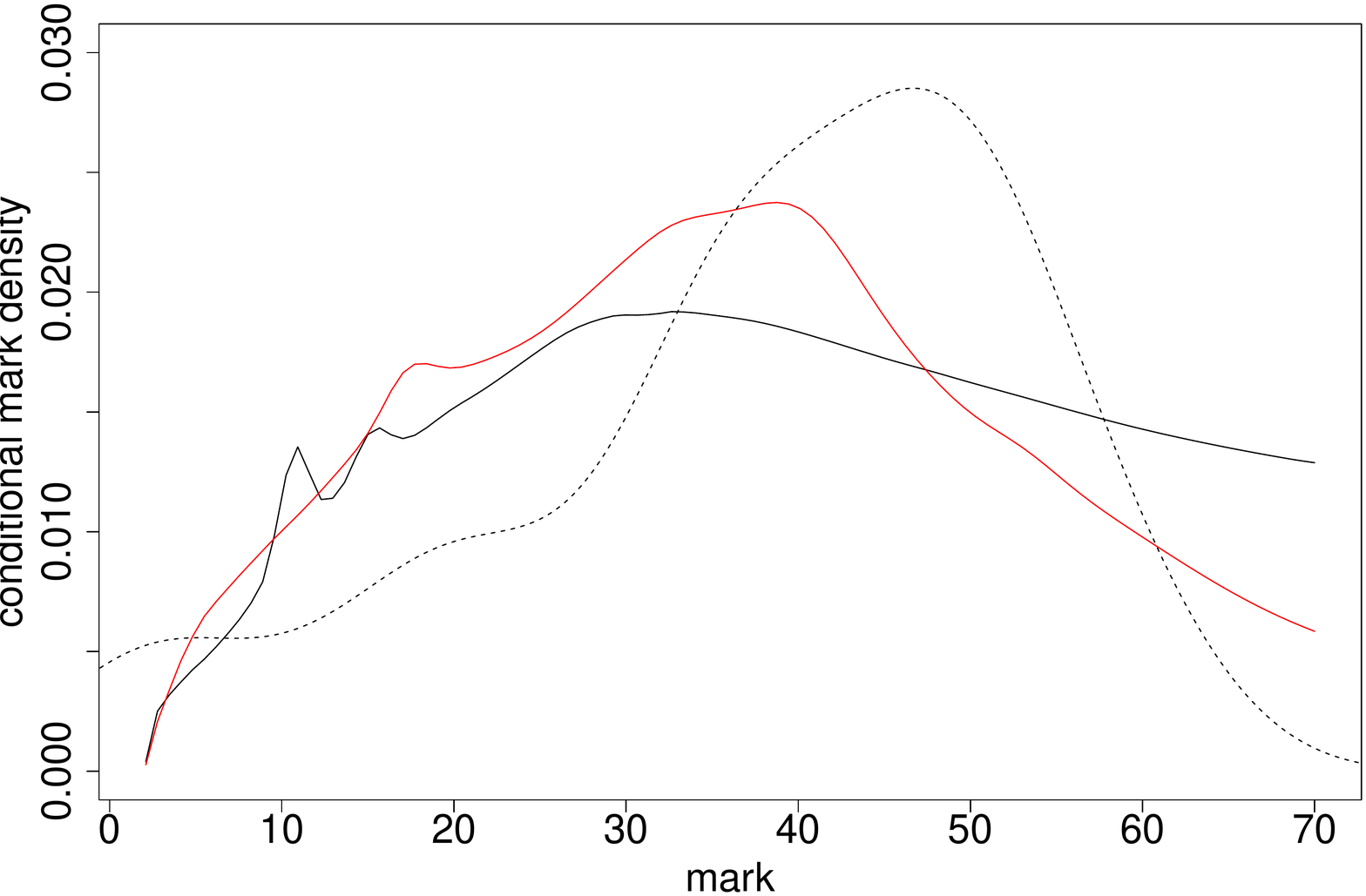}
    }
    \hspace{0mm}
    \subfloat[\centering $s=(105, 140)$]{
    \includegraphics[width=7.5 cm]{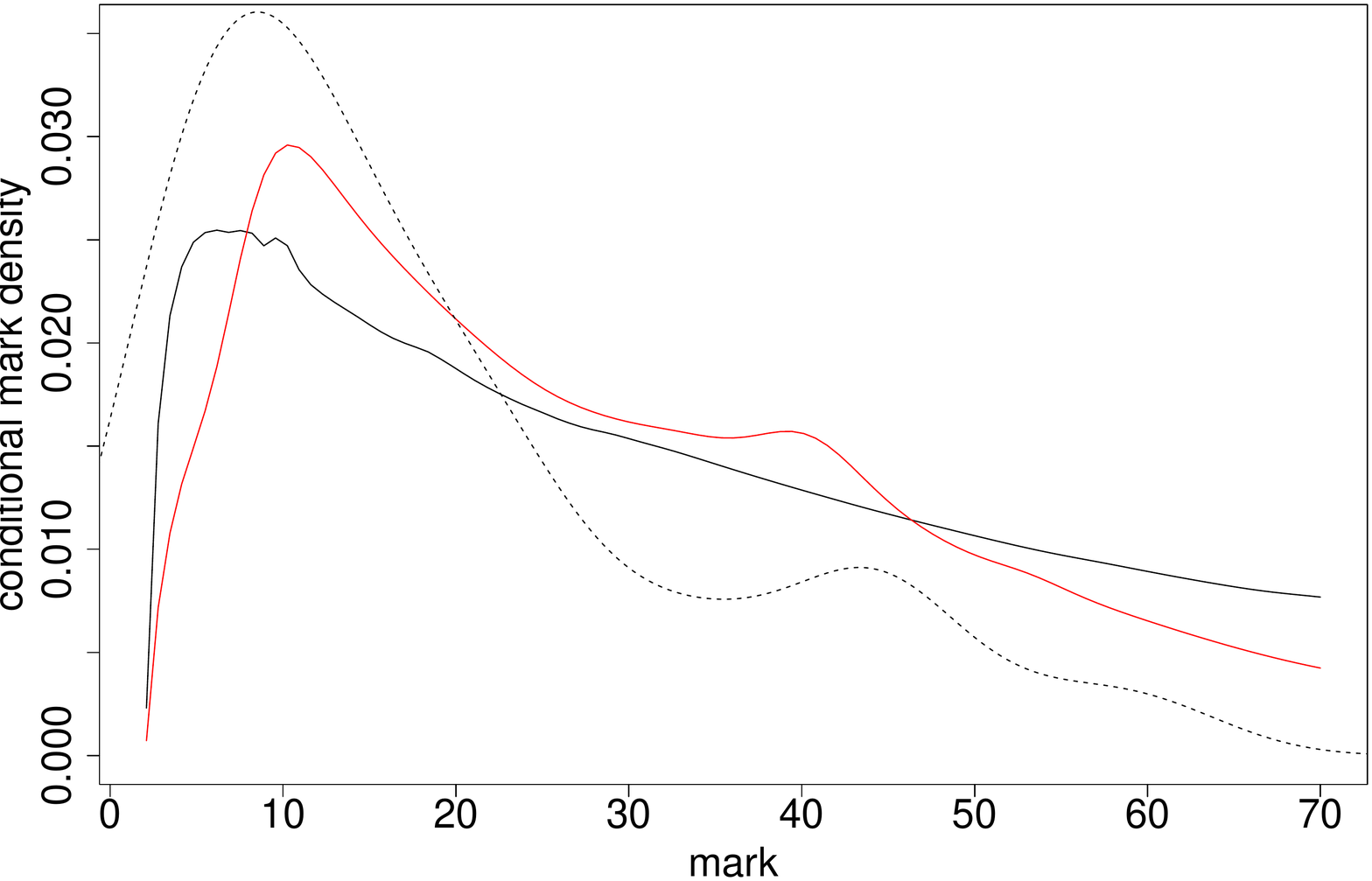}
    }
    \subfloat[\centering $s=(185, 87)$]{
    \includegraphics[width=7.5 cm]{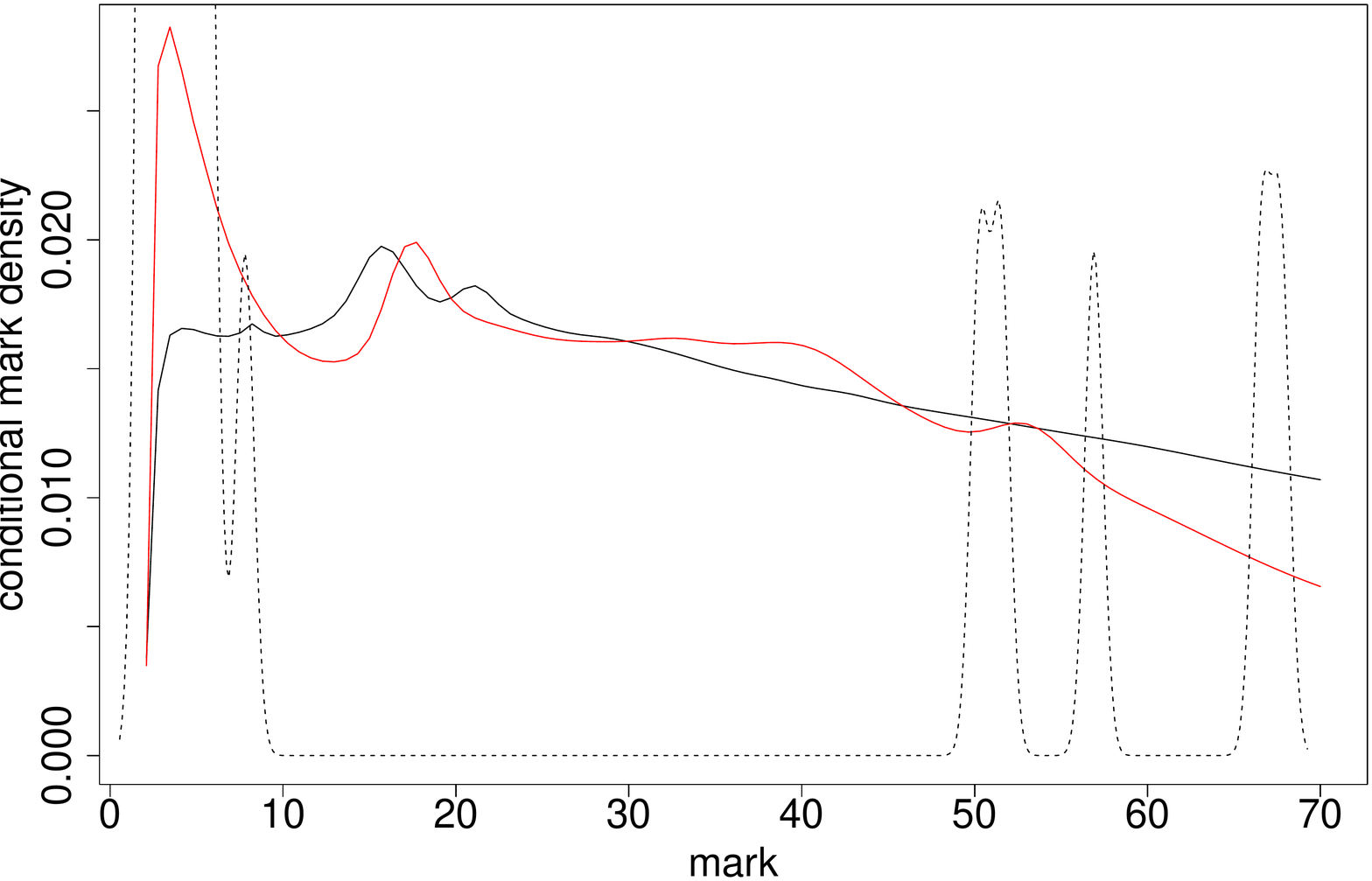}
    }
    \caption{Conditional density estimates for the marks, i.e., the diameter of trees in the {\tt longleaf} dataset at four specific locations, $s$, in the $200 \times 200$ grid with full nine-dimensional (black), reduced six-dimensional (red) mixture model overlayed with an empirical distribution of marks in the neighborhood (dashed)}
    \label{fig:marked}
\end{figure}

% \begin{figure}
% \centering

% \sbox{\bigpicturebox}{%
%   \begin{subfigure}[b]{.45\textwidth}
%   \scalebox{1}[1.2]{\includegraphics[width=\textwidth]{figs/scatter.plot.jpeg}}%
% \caption{Tree locations on a $200 \times 200$ grid, where each point is proportional to the respective tree diameter}
% \label{fig:tree.scatter}
% \end{subfigure}
% }

% \usebox{\bigpicturebox}\hfill
% \begin{minipage}[b][\ht\bigpicturebox][s]{.45\textwidth}
% \begin{subfigure}{.45\textwidth}
% \includegraphics[width=\textwidth]{figs/plot.50.100.pdf}
% \caption{$s=(50, 100)$}
% \label{fig:loc50.100}
% \end{subfigure}\hfill
% \begin{subfigure}{.45\textwidth}
% \includegraphics[width=\textwidth]{figs/plot.81.120.pdf}
% \caption{$s=(81, 120)$}
% \label{fig:loc81.120}
% \end{subfigure}

% \vfill

% \begin{subfigure}{.45\textwidth}
% \includegraphics[width=\textwidth]{figs/plot.105.140.pdf}
% \caption{$s=(105, 140)$}
% \label{fig:loc105.140}
% \end{subfigure}\hfill
% \begin{subfigure}{.45\textwidth}
% \includegraphics[width=\textwidth]{figs/plot.185.87.pdf}
% \caption{$s=(185, 87)$}
% \label{fig:loc185.87}
% \end{subfigure}
% \end{minipage}
% \caption{Scatter plot in (a) with {\em yellow} triangles indicating 4 specific locations at which the conditional mark density is estimated in (b), (c), (d), (e) with nine-dimensional model (black) and six-dimensional model (red). {\color{magenta} The density estimate plots on the right are too small, split this up into two figures: one with the scatterplot only, the other with a $2 \times 2$ layout of the four density estimate plots.}}
% \label{fig:marked}
% \end{figure}

\section{Conclusion}
\label{S:discuss}

% In this paper we show that the PR algorithm is practically feasible in the case of a multivariate mixing distribution using the PRticle filter approach we proposed.  {\color{red} Talk about why this is an important development and what new opportunities it creates.  You could also mention briefly some limitations (e.g., in terms of how high the dimension can be) or other open questions...} 

In this paper we proposed a new filtering mechanism, a PRticle filter, for fitting nonparametric mixture models using the PR algorithm in multivariate problems.  This new development is an important addition because, previously, the PR algorithm could only handle mixtures over relatively low-dimensional spaces.  This contribution creates new opportunities for PR-based methodology in non-trivial problems like marked spatial point process modelling in Section~\ref{SS:marks}. Theoretically, we show that the PRticle filter approximation of the mixing distribution converges to the PR estimate in a strong sense as the number of particles $T$ goes to infinity, when the data $X_1,\ldots,X_n$ of size $n$ remains fixed. This holds for the primary PR run only, an analysis of the attrition-handling embellishments in Section~\ref{attrition} would require more sophisticated techniques.  Coupling this with results in literature on consistency (as $n \to \infty$) of the PR estimator strengthens both the theoretical and practical aspects of PR. Our numerical results show that the PRticle filter approximation gives as accurate results as the traditional PR approach for univariate and bivariate mixtures and is also effective in estimating a multivariate mixture density. 

One might also be interested in quantifying uncertainty about the mixing distribution and its features, like in Section~\ref{SS:marks}. Capturing the variability in the PR estimate is a difficult problem, but suggestions have been made in \citet{fortinipetrone2020} and \citet{dixitmartin2019}. The former uses a quasi-Bayes strategy to construct credible intervals for the PR estimate, while in the latter we leverage the order dependence of the PR estimator for uncertainty quantification. This strategy, which constructs multiple PR estimates based on distinct permutation of the data sequence, would be applicable for the PRticle filter approximation. There are some theoretical gaps that need to be filled, however, so remains an ongoing work. 

One of our numerical illustrations considered nonparametric estimation of mixing distributions supported on the sphere in three-dimensions.  A natural question is if this approach could be extended to other cases involving mixture defined on more general compact manifolds, e.g., higher-dimensional spheres, tori, etc. All that would be needed to extend the proposed strategy in such cases is a map from the surface of the manifold to an underlying Euclidean space where the comptutations can be carried out.  In the special case of the sphere, there is a ``global'' Euclidean-space representation but, for more general manifolds, the corresponding Euclidean spaces would be ``local,'' which creates some new and interesting conceptual and computational challenges.  

%Similarly, this method can be effectively used to model directional data with mixtures. We show an implementation for directional data in $\RR^2$, however that may not be a restriction. The same approach can be extended for mixtures on compact manifolds as the PR algorithm itself is not restricted to Euclidean space and the computational flexibility offered by the PRticle filter approximation makes this extension possible. 

%We also show an application of the multivariate implementation of PR for a real dataset. In this we look at nonparametric mixture modeling of the intensity function of a non-homogeneous Poisson process for marked point processes.

An interesting theoretical question is if consistency of the PRticle filter approximation could be established.  That is, if $\hat P_{n,T}$ is the PRticle filter approximation of the PR estimator $P_n$, then the goal would be to show that $\hat P_{n,T} \to P$ as both $n$ and $T$ go to infinity.  Of course, this would require $T=T_n$ to be increasing sufficiently fast with $n$.  Direct extension of the argument used in the proof of Theorem~\ref{thm:limit} may be possible using some naive techniques, e.g., the classical union bound, but, if successful, this would require $T$ to be exponentially large with $n$.  Our gut feeling is that such a large number of particles would not be necessary, so some important insights are still missing.  We save this as a topic for future work.  

%However one might also be interested in a simultaneous convergence of the PRticle filter approximation to the truth when both the particle size $T$ and sample size $n$ go to infinity. We are yet to investigate this and propose that as future work.

A remaining practical challenge is the handling of attrition when the dimension of the mixing distribution support is relatively high.  What we proposed in Section~\ref{attrition} is able to adequately control attrition rates for mixtures over at least nine-dimensional spaces.  We have not thoroughly tested the performance of the PRticle filter approximation in dimensions higher than this, but we fully expect that controlling the attrition rate will be more and more difficult as the dimension increases.  This is not a limitation of the proposed method, it is a challenge that any importance sampling-based method will face in high-dimensional applications.  New insights would be needed to make this leap to high-dimensional mixtures but it may be possible to take advantage of the PR-specific recursive structure that we used to develop the PRticle filter approximation here.  

%A slight limitation of the PRticle filter in comparison to the traditional PR approach is that, in the former we obtain a particle approximation of the mixing distribution function while the latter gives a proper density estimate. However, the chosen filter is in our control and with a good approximation we can extract all features of the unknown $P$, just like the traditional PR approach. Moreover, estimation of $P$ under the original approach is really not feasible for more than two variables, so in this way a particle approximation is more than what we had before. The idea here is that, if a quadrature scheme is possible and computationally effective then it will give us a density estimate of $P$ which could be preferable in some situations. Another guess we have is that, even though technically applicable for any multivariate mixture model, the PRticle filter would do well for still a small number of variables. This is because we are approximating a nonparametric mixing distribution defined over multiple variables by a finite size of particles. Knowing no other information about this distribution means that as the number of variables increase, this unknown space increases by a dimension and a finite size of particles will not be able to catch all features of the unknown $P$. The attrition resulting from this would be much severe and difficult to handle. More sophisticated attrition reducing techniques would be required as the number of variables increase.

\section*{Acknowledgments}

The authors thank the three anonymous reviewers for their helpful feedback on an earlier version of this manuscript.  This work was supported by the U.S.~National Science Foundation, grant DMS--1737929.

\appendix

\section{Proof of Theorem~\ref{thm:limit}}
\label{proofs}

% {\color{magenta} Maybe we don't need this recollection of the notation...} Recall that $\delta_i(u;X_{i+1}) $ is given by 
% \[ \delta_i(u;X_{i+1}) = 1 + w_i \Bigl( \frac{k(X_{i+1} \mid u)}{m_i(X_{i+1})} - 1 \Bigr), \quad i=0,1,2,\ldots \]
% Note that this depends implicitly on all the data $X^{i+1} = (X_1,\ldots,X_{i+1})$. The corresponding plug-in version is 
% \[ \hat\delta_i(u) = 1 + w_i \Bigl( \frac{k(X_{i+1} \mid u)}{\hat m_{i,T}(X_{i+1})} - 1 \Bigr), \quad i=0,1,2,\ldots \]
% Also, recall that $\Delta_0(u) \equiv 1$ and 
% \[ \Delta_i(u) = \Delta_{i-1}(u) \, \delta_{i-1}(u), \quad i \geq 1, \]
% and, similarly, the plug-in version is $\hat\Delta_0(u) \equiv 1$ and 
% \[ \hat\Delta_i(u) = \hat\Delta_{i-1}(u) \, \hat\delta_{i-1}(u), \quad i \geq 1. \]
Recall that,
\[ \hat m_{i-1}(X_i) = \frac{1}{T}\sum \limits_{t=1}^{T} k(X_i \mid U_t) \, \hat \Delta_{i}(U_t), \quad i \geq 1, \]
where $\hat \Delta_1(u) \equiv 1$ and 
\begin{align*}
\hat \Delta_{i}(u) & = \hat \Delta_{i-1}(u) \, \hat\delta_{i-2}(u) \\
%=\prod\limits_{j=2}^{i} \delta_{j-2} (u;X_{j-1}) = 
& = \prod\limits_{j=2}^{i} \left \{ 1 + w_{j-1} \left ( \frac{k(X_{j-1} \mid u)}{\hat m_{j-2} (X_{j-1})} - 1 \right ) \right \}, \quad i \geq 2. 
\end{align*}
By the strong law of large numbers, we have that 
\[ \hat m_{0,T}(X_1) = \frac{1}{T} \sum_{t=1}^T k(X_1 \mid U_t) \to m_0(X_1), \quad \text{with $P_0$-probability~1 as $T \to \infty$}. \]
To prove a similar claim for all $\hat m_{\ell,T}(X_{\ell+1})$, we proceed by induction.  That is, we start by assuming that 
\begin{equation}
\label{eq:induction}
\hat m_{i-1,T}(X_i) \to m_{i-1}(X_i) \quad \text{with $P_0$-probability~1, for all $i \leq \ell$}, 
\end{equation}
and then use that assumption, along with the structure of the algorithm, to prove  
\[ \hat m_{\ell,T}(X_{\ell+1}) \to m_\ell(X_{\ell+1}), \quad \text{with $P_0$-probability~1, as $T \to \infty$}. \]
Towards this, we have 
\begin{align*}
\hat m_{\ell,T}(X_{\ell+1}) & = \frac1T \sum_{t=1}^T k(X_{\ell+1} \mid U_t) \, \hat \Delta_{\ell + 1}(U_t) \\
& = \frac1T \sum_{t=1}^T k(X_{\ell+1} \mid U_t) \prod_{i=1}^{\ell} \Bigl\{ (1 - w_i) + w_i \frac{k(X_{i} \mid U_t)}{\hat m_{i-1,T}(X_i)} \Bigr\}.
\end{align*}
The product above can be expanded as
\[ \prod_{i=1}^{\ell} \Bigl\{ (1 - w_i) + w_i \frac{k(X_{i} \mid U_t)}{\hat m_{i-1,T}(X_{i})} \Bigr\} = \sum_{{\cal S}(\ell)} \prod_{j \in {\cal S}(\ell)} (1-w_j) \prod_{i \not\in {\cal S}(\ell)} w_i \frac{k(X_{i} \mid U_t)}{\hat m_{i-1,T}(X_{i})}, \]
where ${\cal S}(\ell)$ is a generic subset of $\{1,\ldots,\ell\}$ and the sums and products are over all $2^\ell$ such subsets. Going back the formula for $\hat m_{\ell,T}(X_{\ell+1})$, we can distribute the average over $t$ through the product, which gives 
\[ \hat m_{\ell,T}(X_{\ell+1}) = \sum_{{\cal S}(\ell)} \frac{\prod_{j \in {\cal S}(\ell)} (1-w_j) \prod_{i \not\in {\cal S}(\ell)} w_i}{\prod_{i \not\in {\cal S}(\ell)} \hat m_{i-1,T}(X_{i})} \, \Bigl\{ \frac1T \sum_{t=1}^T k(X_{\ell+1} \mid U_t) \prod_{i \not\in {\cal S}(\ell)} k(X_{i} \mid U_t) \Bigr\}. \]
By the induction hypothesis \eqref{eq:induction}, we have that 
\[ \prod_{i \not\in {\cal S}(\ell)} \hat m_{i-1,T}(X_{i}) \to \prod_{i \not\in {\cal S}(\ell)} m_{i-1}(X_{i}), \quad \text{with $P_0$-probability 1, uniformly in ${\cal S}(\ell)$}. \]
Moreover, by the assumption \eqref{eq:kernel.assumption}, the strong law of large numbers gives 
\[ \frac1T \sum_{t=1}^T k(X_{\ell+1} \mid U_t) \prod_{i \not\in {\cal S}(\ell)} k(X_{i} \mid U_t) \to \int k(X_{\ell+1} \mid u) \prod_{i \not\in {\cal S}(\ell)} k(X_{i} \mid u) \, P_0(du), \]
with $P_0$-probability~1, as $T \to \infty$, again uniformly in ${\cal S}(\ell)$.  The two ``uniformly in ${\cal S}(\ell)$'' claims above follow because there are only finitely many such ${\cal S}(\ell)$.  Putting everything together, we have that $\hat m_{\ell,T}(X_{\ell+1})$ converges with $P_0$-probability~1, as $T \to \infty$, to 
\[ \sum_{{\cal S}(\ell)} \frac{\prod_{j \in {\cal S}(\ell)} (1-w_j) \prod_{i \not\in {\cal S}(\ell)} w_i}{\prod_{i \not\in {\cal S}(\ell)} m_{i-1}(X_{i})} \, \Bigl\{ \int k(X_{\ell+1} \mid u) \prod_{i \not\in {\cal S}(\ell)} k(X_{i} \mid u) \, P_0(du) \Bigr\}. \]
Moving the integration over $u$ to the outside of the sum over ${\cal S}(\ell)$ and undoing the product expansion above eventually leads to $\hat m_{\ell,T}(X_{\ell+1}) \to m_\ell(X_{\ell+1})$ with $P_0$-probability~1.  

We showed above that 
\[ \hat m_{i-1, T}(X_i) \to m_{i-1}(X_i) \quad \text{with $P_0$-probability~1 as $T \to \infty$}, \]
uniformly in $i=1,\ldots,n$ without any assumptions on the convergence of the mixing distribution approximation.  Since the final mixing density estimator $\hat p_{n,T}$ is a continuous function of $\{\hat m_{i-1,T}(X_i): i=1,\ldots,n\}$, it follows that 
\[ \hat p_{n,T}(u) \to p_n(u), \quad \text{with $P_0$-probability~1, as $T \to \infty$, for all $u$}. \]
Since these are density functions, it follows from Scheff\'e's theorem that $\hat p_{n,T}$ converges in $L_1(du)$ to $p_n$, with $P_0$-probability~1, as $T \to \infty$.

\bibliographystyle{apalike}
\bibliography{ref}

\end{document}